%% file: main.tex
\documentclass[journal]{IEEEtran}

\usepackage{lipsum} 

\usepackage{comment}

\usepackage{cite}

\ifCLASSINFOpdf
  \usepackage[pdftex]{graphicx}
\else
\fi

\usepackage{color}

\usepackage{amsmath,amssymb,amsthm,mathtools}
\usepackage{bm}
\usepackage{bbm}
\usepackage{mathrsfs}
\usepackage{upgreek}
\usepackage{soul}
\usepackage{setspace}
\usepackage{cuted}
\usepackage{stfloats}
\usepackage{float}
\usepackage{colortbl}
\usepackage{algcompatible}
\usepackage{algorithm}
\usepackage{algpseudocode}
\algrenewcommand\algorithmicindent{0em}
\makeatletter
\algnewcommand{\LineComment}[1]{\Statex \hskip\ALG@thistlm #1}

\makeatother

\usepackage{slashbox}

\ifCLASSOPTIONcompsoc
  \usepackage[caption=false,font=normalsize,labelfont=sf,textfont=sf]{subfig}
\else
  \usepackage[caption=false,font=footnotesize]{subfig}
\fi

\usepackage[acronym,shortcuts]{glossaries}

\newacronym{5G}{5G}{fifth generation}
\newacronym{6G}{6G}{sixth generation}
\newacronym{ADR}{ADR}{antenna decentralized rate}
\newacronym{AER}{AER}{activity error rate}
\newacronym{AP}{AP}{access point}
\newacronym{ADC}{ADC}{analog-to-digital converter}
\newacronym{ABP}{ABP}{approximate BP}
\newacronym{ADD}{ADD}{annealed discrete denoiser}
\newacronym{AMP}{AMP}{approximate message passing}
\newacronym{AUD}{AUD}{active user detection}
\newacronym{AWGN}{AWGN}{additive white Gaussian noise}
\newacronym{BMMSE}{BMMSE}{Bussgang minimum mean square error}
\newacronym{BC}{BC}{belief combining}
\newacronym{BER}{BER}{bit error rate}
\newacronym{BiGAMP}{BiGAMP}{bilinear generalized approximate message passing}
\newacronym{BIP}{BIP}{bilinear inference problem}
\newacronym{BOD}{BOD}{Bayes-optimal denoiser}
\newacronym{GAMP}{GAMP}{generalized approximate message passing}
\newacronym{GF}{GF}{grant-free}
\newacronym{BiGaBP}{BiGaBP}{bilinear Gaussian belief propagation}
\newacronym{BP}{BP}{belief propagation}
\newacronym{BS}{BS}{base station}
\newacronym{CAP}{CAP}{central AP}
\newacronym{CAMP}{CAMP}{convolutional AMP}
\newacronym{CCU}{CCU}{central computing unit}
\newacronym{CDF}{CDF}{cumulative distribution function}
\newacronym{CDMA}{CDMA}{code division multiple access}
\newacronym{CLT}{CLT}{central limit theorem}
\newacronym{CPU}{CPU}{central processing unit}
\newacronym{CE}{CE}{channel estimation}
\newacronym{CF-mMIMO}{CF-mMIMO}{cell-free massive MIMO}
\newacronym{CG}{CG}{conjugate gradient}
\newacronym{CSI}{CSI}{channel state information}
\newacronym{CSIDCO}{CSIDCO}{complex SIDCO}
\newacronym{DCC}{DCC}{dynamic cooperation clustering}
\newacronym{DFT}{DFT}{discrete Fourier transform}
\newacronym{DoF}{DoF}{degrees of freedom}
\newacronym{DQ}{DQ}{De-quantization}
\newacronym{DSM}{DSM}{denoising score matching}
\newacronym{DL}{DL}{deep learning}
\newacronym{DU}{DU}{deep unfolding}
\newacronym{eMBB}{eMBB}{enhanced mobile broadband}
\newacronym{ECF}{ECF}{estimate-compress-forward}
\newacronym{EP}{EP}{expectation propagation}
\newacronym{EPA}{EPA}{approximate EP}
\newacronym{EXIT}{EXIT}{extrinsic information transfer}
\newacronym{FA}{FA}{false alarm}
\newacronym{FN}{FN}{factor node}
\newacronym{FG}{FG}{factor graph}
\newacronym{FTN}{FTN}{Faster-than-Nyquist}
\newacronym{GaBP}{GaBP}{Gaussian belief propagation}
\newacronym{GLM}{GLM}{generalized linear model}
\newacronym{IC}{IC}{interference cancellation}
\newacronym{IDD}{IDD}{iterative detection and decoding}
\newacronym{i.i.d.}{i.i.d.}{independent and identically distributed}
\newacronym{IoT}{IoT}{Internet of Things}
\newacronym{JACDE}{JACDE}{joint activity, channel and data estimation}
\newacronym{JACE}{JACE}{joint activity and channel estimation}
\newacronym{JCDE}{JCDE}{joint channel and data estimation}
\newacronym{KLD}{KLD}{Kullback-Leibler divergence}
\newacronym{LE}{LE}{linear estimator}
\newacronym{LSA}{LSA}{latent semantic analysis}
\newacronym{LMMSE}{LMMSE}{linear MMSE}
\newacronym{LM}{LM}{long memory}
\newacronym{MAC}{MAC}{multiple-access channel}
\newacronym{MAP}{MAP}{maximum \textit{a posteriori} probability}
\newacronym{MAMP}{MAMP}{memory AMP}
\newacronym{MCS}{MCS}{modulation and coding scheme}
\newacronym{MPDQ}{MPDQ}{message passing DQ}
\newacronym{MD}{MD}{miss-detection}
\newacronym{MF}{MF}{matched filter}
\newacronym{MFB}{MFB}{matched filter bound}
\newacronym{MM}{MM}{moment matching}
\newacronym{MNS}{MNS}{minimum norm solution}
\newacronym{MIMO}{MIMO}{multiple-input multiple-output}
\newacronym{MU-MIMO}{MU-MIMO}{multi-user MIMO}
\newacronym{MU}{MU}{multi-user}
\newacronym{mMIMO}{mMIMO}{Massive multiple-input multiple-output}
\newacronym{mMTC}{mMTC}{massive machine-type communications}
\newacronym{ML}{ML}{maximum likelihood}
\newacronym{MMSE}{MMSE}{minimum mean-square error}
\newacronym{MMV-AMP}{MMV-AMP}{multiple measurement vector approximate message passing}
\newacronym{MSE}{MSE}{mean square error}
\newacronym{MUD}{MUD}{multi-user detection}
\newacronym{NLE}{NLE}{nonlinear estimator}
\newacronym{NN}{NN}{neural network}
\newacronym{NR}{NR}{new radio}
\newacronym{NOMA}{NOMA}{non-orthogonal multiple access}
\newacronym{NMSE}{NMSE}{normalized mean square error}
\newacronym{OAMP}{OAMP}{orthogonal approximate message passing}
\newacronym{OFDM}{OFDM}{orthogonal frequency-division multiplexing}
\newacronym{PBD}{PBD}{parameterized Bayesian denoiser}
\newacronym{PDA}{PDA}{probabilistic data association}
\newacronym{PDF}{PDF}{probability density function}
\newacronym{PMF}{PMF}{probability mass function}
\newacronym{PIC}{PIC}{parallel interference cancellation}
\newacronym{PPP}{PPP}{Poisson point process}
\newacronym{PSK}{PSK}{phase-shift keying}
\newacronym{QP}{QP}{quadratic program}
\newacronym{QPSK}{QPSK}{quadrature phase-shift keying}
\newacronym{QAM}{QAM}{quadrature amplitude modulation}
\newacronym{SIDCO}{SIDCO}{sequential iterative decorrelation via convex optimization}
\newacronym{SD}{SD}{sphere decoding}
\newacronym{SE}{SE}{state evolution}
\newacronym{SGA}{SGA}{scalar Gaussian approximation}
\newacronym{SIC}{SIC}{soft interference cancellation}
\newacronym{SID}{SID}{self-iterative detection}
\newacronym{SNR}{SNR}{signal-to-noise ratio}
\newacronym{Soft IC}{Soft IC}{soft interference cancellation}
\newacronym{SotA}{SotA}{state-of-the-art}
\newacronym{SVD}{SVD}{singular value decomposition}
\newacronym{SPA}{SPA}{sum-product algorithm}
\newacronym{TX}{TX}{transmit}
\newacronym{T-OAMP}{T-OAMP}{trainable OAMP}
\newacronym{RX}{RX}{receive}
\newacronym{UAMP}{UAMP}{unitary AMP}
\newacronym{UE}{UE}{user equipment}
\newacronym{ULA}{ULA}{uniform linear array}
\newacronym{URA}{URA}{unsourced random access}
\newacronym{URLLC}{URLLC}{ultra reliable low latency communications}
\newacronym{VAMP}{VAMP}{vector AMP}
\newacronym{VGA}{VGA}{vector Gaussian approximation}
\newacronym{VN}{VN}{variable node}
\newacronym{w.r.t.}{w.r.t.}{with respect to}
\newacronym{WS}{WS}{warm-started}
\newacronym{ZF}{ZF}{zero-forcing}
\newacronym{flops}{flops}{floating point operations}
\newacronym{CS}{CS}{compressed sensing}
\newacronym{MP}{MP}{message passing}
\newacronym{MPA}{MPA}{message passing algorithm}
\newacronym{DNN}{DNN}{deep neural network}
\newacronym{ASB}{ASB}{adaptively scaled belief}
\newacronym{LLR}{LLR}{log-likelihood ratio}
\newacronym{BAd-VAMP}{BAd-VAMP}{bilinear adaptive vector AMP}
\newacronym{LIP}{LIP}{linear inference problem}
\newacronym{AoA}{AoA}{angle of arrival}
\newacronym{LS}{LS}{least square}
\newacronym{mmWave}{mmWave}{millimeter-wave}

\usepackage{amsthm}

\newtheoremstyle{italicstyle}
  {3pt} 
  {3pt} 
  {\normalfont} 
  {} 
  {\itshape}  
  {:} 
  {0.5em}  
  {\textit{\thmname{#1} \thmnumber{#2}\thmnote{ (#3)}}}
  
\theoremstyle{italicstyle}
\newtheorem{definition}{Definition}

\newtheorem{proposition}{Proposition}
\newtheorem{lemma}{Lemma}

\begin{document}
%

\title{
State-Evolution-based Score Matching\\
for Generalized Approximate Message Passing
}


\author{
Tomoharu~Furudoi,~\IEEEmembership{Graduate Student Member,~IEEE},
~Takumi~Takahashi,~\IEEEmembership{Member,~IEEE},\\
and~Hideki~Ochiai,~\IEEEmembership{Fellow,~IEEE}
\thanks{

The authors are with Graduate School of Engineering, The University of Osaka, 2-1 Yamada-oka, Suita, 565-0871, Japan (e-mail: furutomo@wcs.comm.eng.osaka-u.ac.jp, \{takahashi, ochiai\}@comm.eng.osaka-u.ac.jp).}
\vspace{-7mm}
}

%



\IEEEtitleabstractindextext{%
\vspace{-3ex}
\begin{abstract}
Generalized approximate message passing (GAMP) equipped with minimum mean-square error (MMSE) denoisers, commonly referred to as \emph{Bayes-GAMP}, is a powerful framework for solving inverse problems described by generalized linear models (GLMs) with arbitrary component-wise nonlinearities in the observation process.
However, despite its theoretical tractability and rigorously established asymptotic optimality, the range of practical observation models for which Bayes-GAMP admits a closed-form implementation remains severely limited, particularly in complex-valued settings.
This limitation largely stems from the restrictive requirement that the corresponding output denoiser, given by a conditional expectation, admit a closed-form expression.
To overcome this limitation, we propose a principled approach that enables the implementation of Bayes-GAMP for complex-valued models with \emph{virtually arbitrary} nonlinear observation mappings.
Specifically, within a score-matching framework, we train a neural network to emulate the output denoiser using training data generated from a characterization of the message dynamics based on state evolution (SE).
Notably, the proposed approach requires neither explicit evaluation of the denoiser nor knowledge of an explicit functional form of the nonlinear mapping; it requires only access to forward evaluations of the mapping during offline training.
We show that, under ideal training conditions, GAMP with the trained network replacing the analytically intractable denoiser asymptotically matches the performance of Bayes-GAMP with the exact denoiser.
%
\end{abstract}
}

\maketitle

\IEEEdisplaynontitleabstractindextext

%
\IEEEpeerreviewmaketitle

\glsresetall

\vspace{-3mm}

\section{Introduction}
\label{chap: intro}
\input{TXT/intro}

\section{Preliminaries}
\input{TXT/prelim}

\section{Generalized Approximate Message Passing}
\input{TXT/gamp}

\section{SE-DSM for Output Denoiser Emulation}
\label{chap: se}
\input{TXT/se_dsm}

\vspace{-1ex}
\section{Proof Sketch of Proposition~\ref{prop: se_dsm}}
\label{chap: numerical_result}
\input{TXT/proof}

\section{Conclusion}
\input{TXT/conc}

\section*{Acknowledgment}
The authors thank Professor Tadashi Wadayama of Nagoya Institute of Technology for valuable discussions on score-based inference and related topics, which provided useful background for this work.
This work was supported in part by JSPS KAKENHI Grant Numbers JP26K00949 and JP25H01111.
The first author was supported by JSPS KAKENHI Grant Number JP26KJ1668.

\appendices
\input{TXT/appendix}


\bibliographystyle{REF/IEEEtran}
\bibliography{REF/IEEEabrv,REF/conf_abbrv,REF/ref}

\end{document}

%% file: TXT/intro.tex
Consider the problem of estimating the signal $\bm{x} \in \mathbb{C}^{M \times 1}$ under the following \ac{GLM}~\cite{Pace1997book}:
\begin{equation}
    \label{equ: glm}
    \bm{z} = \bm{A} \bm{x} \in \mathbb{C}^{N \times 1},
    \quad
    \bm{y} = h(\bm{z}, \bm{w}),
\end{equation}
where $\bm{y} \in \mathbb{C}^{N \times 1}$ is the observation vector, $\bm{A} \in \mathbb{C}^{N \times M}$ is the known measurement matrix,
$h: \mathbb{C}^2 \to \mathbb{C}$ represents a nonlinear observation mapping\footnote{
Throughout this paper, a scalar-valued function applied to vector-valued arguments is understood to act component-wise.
Specifically, for
$f:\mathbb{C}^{L}\rightarrow\mathbb{C}$
and vectors
$\bm{c}_{1},\bm{c}_{2},\ldots,\bm{c}_{L}\in\mathbb{C}^{K\times1}$,
we define
$
[f(\bm{c}_{1},\bm{c}_{2},\ldots,\bm{c}_{L})]_k
=
f\!\left(
[\bm{c}_{1}]_k,
[\bm{c}_{2}]_k,
\ldots,
[\bm{c}_{L}]_k
\right),
\,\, k=1,2,\ldots,K.
$
},
and $\bm{w} \in \mathbb{C}^{N \times 1}$ is an unknown noise vector.
For simplicity, $\bm{x}$ and $\bm{w}$ are each assumed to have \ac{i.i.d.} entries throughout this paper.

Given perfect knowledge of the nonlinear mapping $h$ and the probability distributions of $\bm{x}$ and $\bm{w}$, the \ac{MMSE} estimate of $\bm{x}$ can, in principle, be obtained as its posterior mean.
However, direct evaluation of the posterior mean generally entails high-dimensional integration and is computationally prohibitive for large $N$ and $M$.
%
%
%
In this context, \ac{BP}~\cite{Pearl1982, Kschischang2001} provides a general framework for approximate Bayesian inference.
\Ac{GAMP}~\cite{Rangan2011, Rangan2012arxiv} is a low-complexity approximation to \ac{BP} with a per-iteration computational complexity of order $\mathcal{O}(MN)$.
When equipped with \ac{MMSE} denoisers, \ac{GAMP}, often referred to as \emph{Bayes-\ac{GAMP}}, is particularly appealing because it has been rigorously proven to asymptotically achieve Bayes-optimal (\ac{MMSE}) performance in the large-system limit\footnote{
The large-system limit refers to the asymptotic regime of \eqref{equ: glm} in which $N$ and $M$ tend to infinity while their ratio $\xi \triangleq N/M$ remains fixed.} under suitable technical conditions.

Despite its theoretical tractability and the generality of the \ac{GLM} setting, closed-form implementations of the corresponding Bayes-\ac{GAMP} update rules are feasible only in limited scenarios in practice.
In many cases, this difficulty stems from denoiser design.
Specifically, \ac{GAMP} for \acp{GLM} uses two types of denoisers: an input denoiser and an output denoiser.
These denoisers iteratively compute component-wise conditional expectations associated with the input and output signals $\bm{x}$ and $\bm{z}$, respectively.
Of the two, the output denoiser accounts for the nonlinear observation model and constitutes the main analytical bottleneck addressed in this paper.
For $n=1,2,\ldots,N$, let $z_n$, $y_n$, and $w_n$ denote the $n$-th entries of $\bm{z}$, $\bm{y}$, and $\bm{w}$, respectively.
More specifically, closed-form evaluation of the Bayes-optimal output denoiser generally requires both of the following conditions:
\begin{enumerate}
    \item[(C1)] The likelihood $p_{\mathsf{y}_n \mid \mathsf{z}_n}(y_n \mid \cdot)$ induced by the nonlinear mapping $y_n = h(z_n, w_n)$ is analytically available.
    \item[(C2)] The convolution of $p_{\mathsf{y}_n \mid \mathsf{z}_n}(y_n \mid \cdot)$ with a noncentral Gaussian \ac{PDF} over $z_n$ admits a closed-form expression.
\end{enumerate}
%
%
Except for a few canonical cases, such as quantization~\cite{Kamilov2012, Wen2016, Yang2020, Takeuchi2025}, clipping~\cite{Zhidkov2019vtc, Yang2021, Chi2024}, and phase retrieval~\cite{Schniter2015, Zhu2019spl, Ma2019TIT}, (C1) and (C2) are generally difficult to satisfy.
For Bayes-\ac{GAMP} in the complex domain, (C2) poses an even more stringent challenge.
Indeed, when $h$ acts non-separably on the real and imaginary parts, as in amplitude clipping in the complex plane~\cite{Zhidkov2019vtc, Yang2021}, one must resort to either denoiser approximation~\cite{Yang2021} or numerical integration~\cite{Zhidkov2019vtc}.
Such remedies, however, are not always practical and tend to sacrifice either the low computational complexity or the estimation accuracy of \ac{GAMP}.
This limitation is particularly restrictive in communication applications, where amplitude and phase serve as fundamental degrees of freedom.

The primary objective of this paper is to provide a simple and principled procedure for training a \ac{NN} to emulate the behavior of the output denoiser, even when its closed-form expression is unavailable.
This enables, in principle, the implementation of Bayes-\ac{GAMP} for \emph{virtually arbitrary} nonlinear mappings $h$, without being constrained by the availability of a closed-form denoiser expression.
The proposed strategy consists of three steps:
\begin{enumerate}
    \item[(S1)] Express the Bayes-\ac{GAMP} output denoiser in terms of a score function.
    \item[(S2)] Train an \ac{NN} using a score-matching objective~\cite{Hyvarinen2005, Vincent2011} to approximate the target score function.
    \item[(S3)] Use the trained network as a surrogate output denoiser in \ac{GAMP}, thereby avoiding the need for an explicit expression for the likelihood $p_{\mathsf{y}_n \mid \mathsf{z}_n}(y_n \mid \cdot)$ or evaluation of its associated convolution integrals.
\end{enumerate}
%
Although conceptually straightforward, this strategy faces a fundamental challenge: the loss function in the score-matching step (S2), defined as the \ac{MSE} between the \ac{NN} output and the target score, depends on the distribution of the messages propagated by \ac{GAMP}. 
Consequently, training the \ac{NN} requires an accurate characterization of the statistics of the messages that Bayes-\ac{GAMP} would generate with the analytically intractable output denoiser.

A crucial role in addressing this fundamental issue is played by the \emph{\ac{SE}} formalism, a rigorous framework developed for characterizing the asymptotic behavior of several classes of \acp{MPA}~\cite{Donoho2009, Bayati2011}.
Under the \ac{SE} analysis tailored to \ac{GAMP}~\cite{Rangan2011, Rangan2012arxiv, Javanmard2013, Feng2022}, the statistical dynamics of the messages propagated in \ac{GAMP} can be asymptotically characterized by an equivalent scalar system in the sense of \emph{empirical convergence}~\cite{Rangan2012arxiv, Rangan2019}.
This scalar model allows us to generate training data that faithfully reflect the asymptotic statistical structure of the messages, without explicitly running the high-dimensional \ac{GAMP} algorithm itself.

However, another fundamental hurdle still remains:
when the output denoiser is not available in closed form, the corresponding true score function is also intractable; hence, score matching is not directly applicable.
To overcome this issue, we propose a novel reformulation of the loss function in (S2) that enables \ac{NN} training without requiring explicit knowledge of the true score function.
%
Consequently, even though the equivalent scalar system induced by \ac{SE} differs from the standard \ac{AWGN} corruption model typically assumed in the context of score-based generative modeling~\cite{Song2019} and conventional \emph{\ac{DSM}}~\cite{Vincent2011}, one obtains a nontrivial result reminiscent of \ac{DSM}: a surrogate score function can be obtained \emph{simply} by training an \ac{NN} to predict artificially injected Gaussian noise from the generated training data.
Importantly, this training procedure does not require \emph{any} analytical knowledge of the nonlinear mapping $h$ (\textit{e.g.}, its explicit functional form); it only requires forward evaluations of $h$.
These results imply that Bayes-\ac{GAMP} can be implemented for a significantly broader class of nonlinear mappings $h$, possibly including black-box ones, even when the corresponding output denoiser is analytically intractable.
Moreover, given the potential Bayes-optimality of Bayes-\ac{GAMP} in the large-system limit under suitable idealized conditions, the proposed strategy may enable near-optimal signal estimation in much broader scenarios in which conventional implementations of Bayes-\ac{GAMP} have been infeasible.

The main contributions of this paper are summarized as follows.
\begin{itemize}
\item We propose a principled framework for training an \ac{NN} that emulates the output denoiser of Bayes-\ac{GAMP} without requiring a closed-form denoiser expression or even an explicit functional form of the nonlinear mapping $h$.
To this end, we provide a systematic procedure for generating training data based on the \ac{SE} formalism.
This establishes a new use of \ac{SE}, extending it beyond its conventional role as a tool for analyzing the asymptotic convergence behavior of \acp{MPA}.
To the best of our knowledge, this is the first work to leverage \ac{SE} to generate training data for the data-driven emulation of an analytically intractable output denoiser.

\item We develop a novel score-matching-based training methodology for emulating analytically intractable output denoisers.
Specifically, we reformulate the output denoiser in terms of a score function, thereby facilitating the training of an \ac{NN} within the framework of score matching~\cite{Hyvarinen2005}.
By carefully reformulating the resulting loss function, we prove that the training procedure reduces to a regression problem in which the network predicts artificially injected Gaussian noise, rather than the true score function.
Although this result bears a resemblance to the \ac{DSM} framework~\cite{Vincent2011}, it is by no means obvious because the training data generated through \ac{SE} do not follow the standard scalar \ac{AWGN} model assumed in conventional \ac{DSM}.
We refer to this training framework as \emph{\ac{SE}-based \ac{DSM} (SE-DSM)}.

\item We prove that, under ideal SE-DSM training, \ac{GAMP} equipped with the trained \ac{NN} asymptotically achieves the same performance as Bayes-\ac{GAMP} equipped with the exact but analytically intractable output denoiser.
Since SE-DSM training is performed entirely offline, the proposed framework offers a practical route toward efficiently implementing Bayes-\ac{GAMP} for a \emph{virtually arbitrary} nonlinear mapping $h$ while largely preserving the appealing practical advantages of \ac{GAMP}.
\end{itemize}

\subsection{Related Works}

The studies in~\cite{Wadayama2026scvamp, Wadayama2026tmvamp} adopt a strategy closely related to that of the present paper.
Specifically, they consider signal estimation problems based on a special class of real-valued \acp{GLM}
(\textit{i.e.}, $\bm{y} = h(\bm{A}\bm{x}) + \bm{w}$, where $h: \mathbb{R} \to \mathbb{R}$ denotes a nonlinear function and $\bm{w}$ is a Gaussian noise vector) and solve them using a score-function-based reformulation of a variant of \ac{VAMP}~\cite{Schniter2016, Rangan2019}.
The \ac{NN} training in these works builds upon a heuristic assumption that the messages propagated to the output denoiser can be characterized by a standard \ac{AWGN} model.
Under this assumption, the classical result of \ac{DSM}~\cite{Vincent2011} can be directly applied.
Consequently, \acp{NN} can be trained even when the true score function, or equivalently the corresponding denoiser, is analytically unavailable.

However, the theoretical validity of this heuristic assumption is not established in~\cite{Wadayama2026scvamp, Wadayama2026tmvamp}.
In fact, the same assumption does not generally hold for the output denoiser of \ac{GAMP}, resulting in a mismatch between the statistical structure of the propagated messages and that of the generated training data.
As a result, the learned mapping may no longer correspond to the true score function, leading to performance degradation when it is plugged into \ac{GAMP}.
In contrast, the proposed SE-DSM framework leverages the \ac{SE} formalism to generate training data that accurately reflect the asymptotic statistical structure of the propagated messages.

On the other hand, score-based \acp{MPA} have also been investigated in~\cite{Yu2024, Cai2025spawc}, particularly in the context of \ac{OAMP}~\cite{Ma2017}.
These works focus on the standard linear model and learn denoisers for input signals with unknown prior distributions using the \ac{DSM} framework.
Consequently, their objective is fundamentally different from that of the present work, which aims to emulate analytically intractable output denoisers and establish a theoretically justified training framework in the more general \ac{GLM} setting.

\textit{Notation}: 
Vectors and matrices are denoted by lowercase and uppercase boldface letters, respectively. 
The sets of nonnegative real, real, and complex numbers are denoted by $\mathbb{R}_+$, $\mathbb{R}$, and $\mathbb{C}$, respectively.
The identity matrix of size $K$ and the zero vector are denoted by $\bm{I}_K$ and $\bm{0}$, respectively.
The $k$-th entry of $\bm{c} \in \mathbb{C}^{K \times 1}$ is denoted by $[\bm{c}]_k$, while the $(k, l)$-th element of $\bm{C} \in \mathbb{C}^{K \times L}$ is denoted by $[\bm{C}]_{k, l}$.
The conjugate, transpose, and conjugate transpose are denoted by $(\cdot)^*$, $(\cdot)^{\mathsf{T}}$, and $(\cdot)^{\mathsf{H}}$, respectively.
The 
imaginary unit is defined as $\mathrm{j} \triangleq \sqrt{-1}$.
We define $\langle \bm{c} \rangle \triangleq \frac{1}{K} \sum_{k=1}^{K} [\bm{c}]_k $ and $\|\bm{c}\|^2 \triangleq \sum_{k=1}^{K} |[\bm{c}]_k|^2$ for all $\bm{c} \in \mathbb{C}^{K \times 1}$.
The circularly symmetric complex Gaussian distribution with mean vector $\bm{\mu}$ and covariance matrix $\bm{\varLambda}$ is denoted by $\mathcal{CN}\left(\bm{\mu}, \bm{\varLambda}\right)$.
The notation $A \sim \mathcal{P}$ denotes that a random variable $A$ follows the distribution $\mathcal{P}$, whereas $A \sim B$ denotes that $A$ and $B$ are identically distributed.
The notation $\{A_k\}_{k=1}^{K} \overset{\mathrm{i.i.d.}}{\sim} A$ means that $A_1, A_2, \ldots, A_K$ are independent copies of $A$.
The density function, with respect to an appropriate reference measure, and expectation for a random variable $A$ are denoted by $p_{\mathsf{A}}(\cdot)$ and $\mathbb{E}_{\mathsf{A}}[\cdot]$, respectively.
The conditional density of $A$ given $B=b$ and the conditional expectation given $B$ are denoted by $p_{\mathsf{A} \mid \mathsf{B}}(\cdot \mid b)$ and $\mathbb{E}_{\mathsf{A}}[\cdot \mid B]$, respectively. %
For notational simplicity, all integrals over $\mathbb{C}$ are understood with respect to the two-dimensional Lebesgue measure, \textit{i.e.},
$
\int_{\mathbb{C}} \cdots \, \mathrm{d}a
=
\int_{\mathbb{R}^2} \cdots \, \mathrm{d}a^{\Re}\mathrm{d}a^{\Im}
$
for $a=a^{\Re}+\mathrm{j}a^{\Im} \in \mathbb{C}$.
Finally, $ \overset{\mathrm{a.s.}}{\to}$ and $\overset{\mathrm{a.s.}}{=}$ indicate almost-sure convergence and equality, respectively.

%% file: TXT/prelim.tex
As preparation for the subsequent analysis, this section introduces several mathematical definitions and technical assumptions associated with the signal model in \eqref{equ: glm}.

For notational convenience, we denote
$a_{n,m} \triangleq [\bm{A}]_{n,m}$,
$x_m \triangleq [\bm{x}]_m$,
$y_n \triangleq [\bm{y}]_n$,
$z_n \triangleq [\bm{z}]_n$,
and
$w_n \triangleq [\bm{w}]_n$
for $n = 1, 2, \ldots,N$ and $m = 1, 2, \ldots,M$.

\subsection{Definitions}

Following \cite{Takeuchi2020, Ma2019TIT}, we first introduce the Wirtinger derivatives with respect to a complex quantity $c \in \mathbb{C}$ and its conjugate $c^*$.

\begin{definition}[Wirtinger Derivative]
For a complex variable $c = c^\Re + \mathrm{j}c^\Im \in \mathbb{C}$ and its conjugate $c^*$, the Wirtinger derivatives are defined by
\begin{equation}
\!
\frac{\partial}{ \partial c}
=
\frac{1}{2}  
\left(\frac{\partial }{ \partial c^\Re} - \mathrm{j}\frac{\partial }{ \partial c^\Im} \right) \!, \,
\frac{\partial }{ \partial c^*}
=
\frac{1}{2}
\left(\frac{\partial }{ \partial c^\Re} + \mathrm{j}\frac{\partial }{ \partial c^\Im} \right).
\end{equation}
Throughout this paper, the notation 
$ f'(c_1, c_2, \ldots , c_L) \triangleq
\frac{\partial}{\partial c_1}  f(c_1, c_2, \ldots , c_L) $ denotes the Wirtinger derivative with respect to the first argument for any differentiable function $f: \mathbb{C}^L \to \mathbb{C}$.
\end{definition}

The following two definitions are useful to formally state the existing results of \ac{SE}-based analysis.

\begin{definition}[Pseudo-Lipschitz Continuity~\cite{Bayati2011, Takeuchi2020}]
A function $f: \mathbb{C}^L \to \mathbb{C}$ is said to be \emph{pseudo-Lipschitz of order $p$} if there exists a constant $C \, (< \infty)$ such that, for any $\bm{r}, \bm{s} \in \mathbb{C}^{L \times 1}$, the following inequality holds.
\begin{equation}
\left| f(\bm{r}) - f(\bm{s}) \right|
\le
C \, \|\bm{r} - \bm{s}\|
\bigl( 1 + \|\bm{r}\|^{p-1} + \|\bm{s}\|^{p-1} \bigr).
\end{equation}
When $p=1$, this definition reduces to that of Lipschitz continuity~\cite{Rangan2019}.
\end{definition}

\begin{definition}[Empirical Convergence~\cite{Rangan2012arxiv, Rangan2019}]
Let $K$ be a positive integer satisfying $\lim_{N = \xi M \to \infty} K/M \in (0, \infty)$, \textit{i.e.}, $K$ grows at the same rate as $N$ and $M$ in the large-system limit $N = \xi M \to \infty$.
Then, a sequence of random variables
\begin{equation}
    \{r_{k, 1}(M), r_{k, 2}(M), \ldots , r_{k, L}(M)\}_{k
    =1}^{K}
\end{equation}
is said to \emph{converge empirically} to a tuple of random variables $(R_1, R_2, \ldots , R_L)$ \emph{in the pseudo-Lipschitz sense of order} $p \, (\geq 2)$ if
\begin{align}
&
\frac{1}{K}
\sum_{k=1}^{K} \phi \left(r_{k, 1}(M), r_{k, 2}(M), \ldots , r_{k, L}(M) \right)
\nonumber
\\
&
\!\!\!\!\!\!\!\!
\overset{\mathrm{a.s.}}{\to}
\mathbb{E}_{ \mathsf{R}_1, \mathsf{R}_2, \ldots ,  \mathsf{R}_L} \left[ \phi(R_1, R_2, \ldots , R_L) \right] < \infty
\label{equ: def_empirical_convergence}
\end{align}
holds almost surely in the large-system limit $N = \xi M \to \infty$ for any pseudo-Lipschitz test function $\phi: \mathbb{C}^{L} \to \mathbb{C}$ of order $p \, (\geq 2)$.
We denote this convergence in this work by
\begin{equation}
\{ r_{k, 1}, r_{k, 2}, \ldots , r_{k, L} \}_{k=1}^{K} \overset{\mathrm{PL}(p)}{ \longrightarrow } (R_1, R_2, \ldots , R_L),
\end{equation}
where the dependence of the random variables on $M$ in \eqref{equ: def_empirical_convergence} is omitted.
\end{definition}

\subsection{Assumptions on the Signal Model}
\label{chap: signal_model_assumption}

To facilitate the subsequent \ac{SE}-based discussions, we introduce the following technical assumptions on the signal model in \eqref{equ: glm}:
\begin{itemize}
\item Each entry of the measurement matrix $\bm{A}$ is independently drawn from $\mathcal{CN}(0, 1/N)$.

\item For simplicity, we assume that there exist independent random variables $X$ and $W$ such that $\{x_m\}_{m=1}^{M} \overset{\mathrm{i.i.d.}}{\sim} X$ and $\{w_n\}_{n=1}^{N} \overset{\mathrm{i.i.d.}}{\sim} W$.
Moreover, for some integer $\ell \geq 2$, both $\mathbb{E} \left[ |X|^{2\ell-2} \right]$ and $\mathbb{E} \left[ |W|^{2\ell-2} \right]$ are finite.

\item We assume that $X$ and $W$ are zero-mean random variables, \textit{i.e.},
$\mathbb{E}_{\mathsf{X}}[X] = 0$ 
and
$\mathbb{E}_{\mathsf{W}}[W] = 0$.
Also, the variance of $X$ is denoted by
$ \sigma_{\mathrm{x}}^2 \triangleq
\mathbb{E}_{\mathsf{X}} \bigl[ |X|^2 \bigr] \, (< \infty) $.
\end{itemize}

%% file: TXT/gamp.tex
This subsection briefly reviews the algorithmic structure of complex-valued Bayes-\ac{GAMP}, which estimates $\bm{x}$ based on the \ac{GLM} in \eqref{equ: glm}.
We then heuristically derive its simplified form in the large-system limit through several approximations.

\vspace{-1ex}
\subsection{Bayes-GAMP}

First, we present the pseudocode of complex-valued Bayes-\ac{GAMP} in Algorithm~\ref{alg: gamp}, which is equivalent to Table~I in \cite{Schniter2015}.
We refer to the $l$-th line of Algorithm~$a$ as (A$a$-$l$) throughout the remainder of this paper, and the notation $(\cdot)(t)$ denotes the value at the $t$-th iteration for $t  = 1, 2, \ldots , T$.

Algorithm~\ref{alg: gamp} consists of two modules, each referred to as \emph{output module} and \emph{input module}, respectively.
The first step of the output module in (A\ref{alg: gamp}-3, 4) is to combine the tentative estimate of $\bm{x}$ and its estimated \ac{MSE} across the input domain $m = 1, 2, \ldots , M$ to yield the \emph{effective} prior information for $z_n$ as
$\left\{\overline{z}_n (t), \overline{v}_n^{\mathrm{z}} (t) \right\}_{n=1}^{N}$.
These variables are used to calculate the tentative estimate of $\bm{z}$ in (A\ref{alg: gamp}-5, 6).
More specifically, the function called \emph{output denoiser} $\eta_{\mathrm{out}}: (\mathbb{C} \times \mathbb{R}_+ \times \mathbb{C}) \to \mathbb{C}$ postulates the scalar \emph{Gaussian-prior model} shown in Fig.~\ref{fig: scalar_models}(a) for a set of $ \left( \overline{z}_n, \overline{v}_n^{\mathrm{z}}, y_n \right) $ independently for $n = 1, 2, \ldots , N$, where
$\overline{z}_n$
and
$\overline {v}_n^{\mathrm z}$
serve as the prior mean and variance, respectively.
Then, it computes the \ac{MMSE} estimate of $z_n$ as
\begin{align}
&
\hat{z}_n =
\eta_{\mathrm{out}} (\overline{z}_n; \overline{v}_n^{\mathrm{z}}, y_n)
\nonumber
\\
&
\!\!\!\!
\triangleq
\frac{1}{p_{\mathsf{y}_n \mid \overline{\mathsf{z}}_n }
(y_n \mid \overline{z}_n)}
\int_{\mathbb{C}} z \cdot p_{\mathsf{y}_n \mid \mathsf{z}_n }(y_n \mid z) 
\,
\frac{1}{\pi \overline{v}_n^{\mathrm{z}}}
e^{-\frac{|\overline{z}_n - z|^2}{\overline{v}_n^{\mathrm{z}}}}
\,
\mathrm{d}z,
\label{equ: def_bo_out_denoiser}
\end{align}
with the normalization factor
\begin{equation}
p_{\mathsf{y}_n \mid \overline{\mathsf{z}}_n }
(y_n \mid \overline{z}_n)
=
\int_{\mathbb{C}}
p_{\mathsf{y}_n \mid \mathsf{z}_n }(y_n \mid z) 
\,
\frac{1}{\pi \overline{v}_n^{\mathrm{z}}}
e^{-\frac{|\overline{z}_n - z|^2}{\overline{v}_n^{\mathrm{z}}}}
\, \mathrm{d}z.
\label{equ: p_y_z_ol}
\end{equation}
Note that (A\ref{alg: gamp}-6) follows from the  identity of $\eta_{\mathrm{out}}$:
\begin{align}
&\overline{v}_n^{\mathrm{z}}
\cdot
\eta_{\mathrm{out}}' (\overline{z}_n; \overline{v}_n^{\mathrm{z}}, y_n) 
\nonumber
\\
&
\!\!\!
= 
\frac{1}{p_{\mathsf{y}_n \mid \overline{\mathsf{z}}_n }
(y_n \mid \overline{z}_n)}
\int_{\mathbb{C}} |z|^2 \cdot p_{\mathsf{y}_n \mid \mathsf{z}_n }(y_n \mid z) 
\,
\frac{1}{\pi \overline{v}_n^{\mathrm{z}}}
e^{-\frac{|\overline{z}_n - z|^2}{\overline{v}_n^{\mathrm{z}}}}
\,
\mathrm{d}z
\nonumber
\\
&
\qquad\qquad\qquad\qquad
\qquad\qquad \quad
-\,|\eta_{\mathrm{out}}(\overline{z}_n; \overline{v}_n^{\mathrm{z}}, y_n)|^2.
\label{equ: eta_out_identity}
\end{align}
Finally, the output module passes messages to the input module through (A\ref{alg: gamp}-7, 8).

\input{ALG/gamp}

Subsequently, the input module first constructs the \emph{effective} likelihood information as $\left\{\overline{x}_m(t), \overline{v}_m^{\mathrm{x}}(t) \right\}_{m=1}^{M}$ in (A\ref{alg: gamp}-9, 10) by combining the propagated messages from the output module across the output domain $n = 1, 2,  \ldots , N$.
Treating these variables as observations generated from the scalar AWGN model shown in Fig.~\ref{fig: scalar_models}(b), the \emph{input denoiser} $\eta_{\mathrm{in}}: (\mathbb{C}\times \mathbb{R}_+) \to \mathbb{C}$ calculates the tentative \ac{MMSE} estimate of $\bm{x}$ independently for $m = 1, 2, \ldots , M$ as
\begin{align}
&
\hat{x}_m =
\eta_{\mathrm{in}}(\overline{x}_m; \overline{v}_m^{\mathrm{x}}) 
\nonumber
\\
&
\quad
\triangleq
\frac{1}{p_{\overline{\mathsf{x}}_m} (\overline{x}_m) }
\int_{\mathbb{C}} x \cdot p_{\mathsf{x}_m}(x)
\,
\frac{1}{\pi \overline{v}_m^{\mathrm{x}}}
e^{-\frac{|\overline{x}_m - x|^2}{\overline{v}_m^{\mathrm{x}}}}
\,
\mathrm{d}x,
\end{align}
with the normalizing factor
\begin{equation}
p_{\overline{\mathsf{x}}_m} (\overline{x}_m) =
\int_{\mathbb{C}} p_{\mathsf{x}_m}(x)
\,
\frac{1}{\pi \overline{v}_m^{\mathrm{x}}}
e^{-\frac{|\overline{x}_m - x|^2}{\overline{v}_m^{\mathrm{x}}}}
\,
\mathrm{d}x.
\end{equation}
Similarly to \eqref{equ: eta_out_identity}, (A\ref{alg: gamp}-12) follows from 
\begin{align}
&
\overline{v}_m^{\mathrm{x}}
\cdot
\eta_{\mathrm{in}}'(\overline{x}_m; \overline{v}_m^{\mathrm{x}}) 
\nonumber
\\
&
\quad
=
\frac{1}{p_{\overline{\mathsf{x}}_m} (\overline{x}_m) }
\int_{\mathbb{C}} |x|^2 \cdot p_{\mathsf{x}_m}(x)
\,
\frac{1}{\pi \overline{v}_m^{\mathrm{x}}}
e^{-\frac{|\overline{x}_m - x|^2}{\overline{v}_m^{\mathrm{x}}}}
\,
\mathrm{d}x
\nonumber
\\
&
\qquad \qquad \qquad \qquad \qquad \qquad \qquad
- \,
|\eta_{\mathrm{in}}(\overline{x}_m; \overline{v}_m^{\mathrm{x}})|^2,
\end{align}
which corresponds to the complex-valued version of equation (32) in \cite{Rangan2012arxiv}.

For a detailed explanation of Algorithm~\ref{alg: gamp}, including its derivation from loopy \ac{BP}, interested readers are referred to \cite{Rangan2011, Rangan2012arxiv, Zou2018, Meng2018, Liu2019tvt}.


\begin{figure}[t]
\centering
    \subfloat[Scalar Gaussian-prior model.]{
        \includegraphics[width=0.9\linewidth]{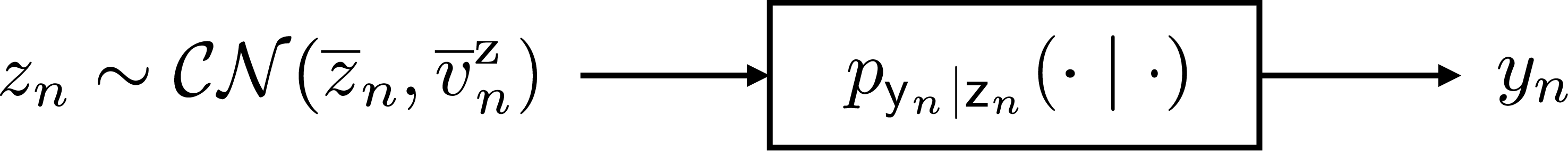}
        \label{fig: gaussian_prior_model}
    }
    \vspace{3mm}
    
    \subfloat[Scalar \ac{AWGN} model.]{
        \includegraphics[width=0.65\linewidth]{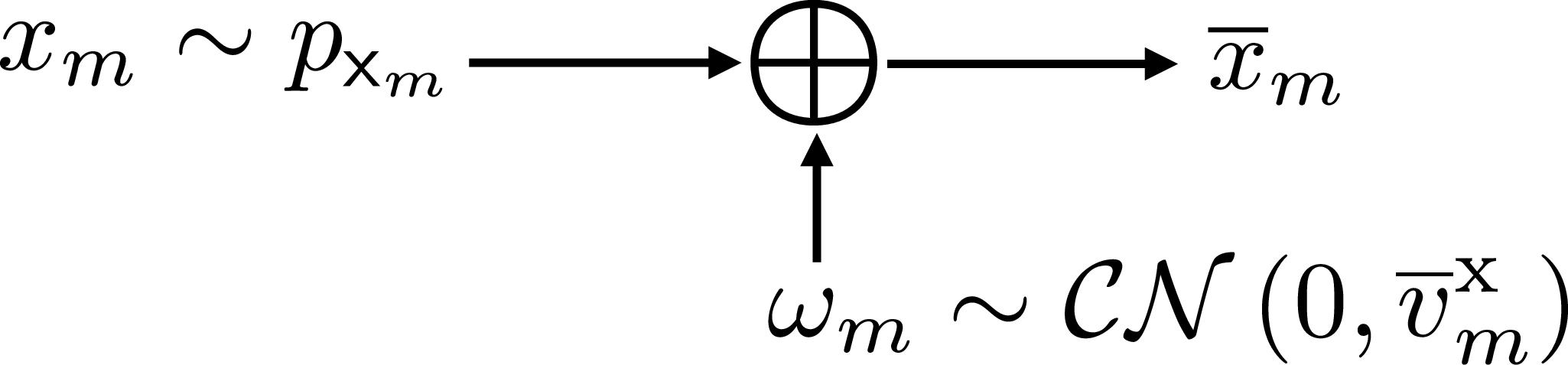}
        \label{fig: awgn_model}
    }
    \vspace{3mm}
    \caption{
    Scalar models based on which the denoisers operate:
    (a) Gaussian-prior model for the output denoiser
    $\eta_{\mathrm{out}}(\overline{z}_n; \overline{v}_{n}^{\mathrm{z}}, y_n)$
    and
    (b) \ac{AWGN} model for the input denoiser
    $\eta_{\mathrm{in}}(\overline{x}_m; \overline{v}_m^{\mathrm{x}})$.
    }
    \label{fig: scalar_models}
\end{figure}

\subsection{Reformulation of Bayes-GAMP via a Score Function}

In what follows, our goal is to simplify Algorithm~\ref{alg: gamp} to facilitate its \ac{SE}-based analysis.
As a first step toward this goal, we reformulate the output denoiser in Algorithm~\ref{alg: gamp} in terms of the \emph{score function} $ \frac{\partial}{\partial \overline{z}^*} \log p_{\mathsf{y}_n \mid \overline{\mathsf{z}}_n} (y \mid \overline{z}) $.

\begin{lemma}[Output Denoiser as a Score Function]
The output denoiser $\eta_{\mathrm{out}}$ defined in \eqref{equ: def_bo_out_denoiser} satisfies the following relationship:
\begin{equation}
\eta_{\mathrm{out}}(\overline{z}; \overline{v}) 
=
\overline{z} + \overline{v} \cdot
\underbrace{
\frac{\partial}{\partial \overline{z}^*}
\log{
p_{\mathsf{y}_n \mid \overline{\mathsf{z}}_n }
(y \mid \overline{z})}}_{\textrm{Score Function}}.
\label{equ: tweedie_output}
\end{equation}
\label{lem: tweedie_output}
\end{lemma}
This identity is a direct extension of the result in [\citenum{Rangan2012arxiv}, Section~IV-B] to complex-valued variables and can be verified by differentiating \eqref{equ: p_y_z_ol} with respect to $\overline{z}^*$.

Lemma~\ref{lem: tweedie_output} indicates that, although the underlying model in Fig.~\ref{fig: scalar_models}(a) is distinct from the standard \ac{AWGN} model, the \emph{Tweedie-type identity}~\cite{Robbins1956, Efron2011} holds for the corresponding \ac{MMSE} estimator $\eta_{\mathrm{out}}$.
Therefore, substituting the following alternative expression of $\hat{z}_n(t)$ in (A\ref{alg: gamp}-4) 
\begin{equation}
    \hat{z}_n(t) = \overline{z}_n(t) + \overline{v}_n^{\mathrm{z}}(t) \cdot
    \frac{\partial}{\partial \overline{z}_n^*}
    \log{
 p_{\mathsf{y}_n \mid \overline{\mathsf{z}}_n }
 (y \mid \overline{z}_n(t))},
\end{equation}
into (A\ref{alg: gamp}-7) yields
\begin{equation}
\check{z}_n(t) = 
\frac{\partial}{\partial \overline{z}_n^*}
\log{
 p_{\mathsf{y}_n \mid \overline{\mathsf{z}}_n }
 (y_n \mid \overline{z}_n(t))}.
 \label{equ: z_check_score}
 \end{equation}
Furthermore, by differentiating both sides of \eqref{equ: tweedie_output}
with respect to $\overline{z}$ and multiplying the result by $\overline{v}$, one obtains
 \begin{equation}
    \overline{v} \cdot 
    \eta_{\mathrm{out}}'(\overline{z}; \overline{v}) 
    =
    \overline{v} + (\overline{v})^2 \cdot 
    \frac{\partial^2}{
    \partial \overline{z}
    \partial \overline{z}^*}
\log{
p_{\mathsf{y}_n \mid \overline{\mathsf{z}}_n }
(y \mid \overline{z})}.
\label{equ: tweedie_output_deriv}
\end{equation}
 Thus, by equating \eqref{equ: tweedie_output_deriv} with (A\ref{alg: gamp}-6) after substituting $ \overline{z} = \overline{z}_n(t), \, \overline{v} = \overline{v}_n^{\mathrm{z}}(t), \, y = y_n $ into \eqref{equ: tweedie_output_deriv}, one arrives at
\begin{equation}
\hat{v}_n^{\mathrm{z}}(t)
=
 \overline{v}_n^{\mathrm{z}}(t) + \left( \overline{v}_n^{\mathrm{z}}(t) \right)^2 \cdot 
\frac{\partial^2}{
\partial \overline{z}
\partial \overline{z}^*}
\log{
p_{\mathsf{y}_n \mid \overline{\mathsf{z}}_n }
(y_n \mid \overline{z}_n(t))},
\nonumber
\end{equation}
 equivalently, we have
 \begin{equation}
\underbrace{
\frac{ 1 - \hat{v}_n^{\mathrm{z}}(t) / \overline{v}_n^{\mathrm{z}}(t) }{\overline{v}_n^{\mathrm{z}}(t)}}_{= 1/\check{v}_n^{\mathrm{z}}(t) \textrm{ from (A\ref{alg: gamp}-8)}}
\!
=
\!
- \frac{\partial^2}{
\partial \overline{z}
\partial \overline{z}^*}
\log{
p_{\mathsf{y}_n \mid \overline{\mathsf{z}}_n }
(y_n \mid \overline{z}_n(t))}.
\label{equ: v_z_check_score}
\end{equation}

\input{ALG/gamp_sc}

Therefore, using \eqref{equ: z_check_score} and \eqref{equ: v_z_check_score},
Algorithm~\ref{alg: gamp} can be equivalently represented as Algorithm~\ref{alg: gamp_sc}, where the input and output \emph{decision functions}, $g_{\mathrm{in}}$ and $g_{\mathrm{out}}$, are defined as
 \begin{subequations}
\begin{align}
&
g_{\mathrm{out}}
 \left( \overline{z}; \overline{v}, y \right)
 \triangleq 
\frac{\partial}{\partial \overline{z}^*}
\log{
p_{\mathsf{y}_n \mid \overline{\mathsf{z}}_n }
(y \mid \overline{z})},
\label{equ: def_g_out}
\\
&
g_{\mathrm{in}}
 \left( \overline{x}; \overline{v} \right)
 \triangleq 
\eta_{\mathrm{in}} (\overline{x}; \overline{v}).
\label{equ: def_g_in}
\end{align}
\end{subequations}

\subsection{Heuristic Approximation in the Large-System Limit}

Finally, assuming the asymptotic regime ($N, M \gg 1$), 
we further simplify Algorithm~\ref{alg: gamp_sc} into a form that is amenable to a direct \ac{SE}-based convergence analysis.

This simplification can be achieved by heuristically averaging out the variance parameters across components.
We begin by averaging $\{ \hat{v}_m^{\mathrm{x}}(t+1) \}_{m=1}^{M}$ in (A\ref{alg: gamp_sc}-10), 
namely, 
\begin{align}
    &
    \hat{v}^{\mathrm{x}} (t+1) \triangleq \frac{1}{M} \sum_{m=1}^{M} \hat{v}_m^{\mathrm{x}}(t+1)
    \nonumber
    \\
    & \quad =
    \frac{1}{M} \sum_{m=1}^{M} 
    \left[\overline{v}_m^{\mathrm{x}}{(t)}
    \cdot
    \frac{\partial}{\partial \overline{x}_m}
    g_{\mathrm{in}} \left(\overline{x}_m{(t)}; \overline{v}_m^{\mathrm{x}}{(t)} \right)\right].
    \label{equ: v_hat_x_ave}
\end{align}%
Then, replacing $\hat{v}_m^{\mathrm{x}} (t)$ with $\hat{v}^{\mathrm{x}}(t)$ in (A\ref{alg: gamp_sc}-3) yields
\begin{align}
&
\overline{v}_n^{\mathrm{z}}(t) 
\to
\sum_{m=1}^{M} |a_{n, m}|^2 \hat{v}^{\mathrm{x}}(t)
\nonumber
\\
&=
M
\hat{v}^{\mathrm{x}}(t)
\cdot 
\underbrace{
\frac{1}{M}
\sum_{m=1}^{M}
|a_{n, m}|^2
}_{\approx \mathbb{E}\left[|a_{n, m}|^2\right] = 1/N}
\approx
\frac{\hat{v}^{\mathrm{x}}(t)}{\xi} \triangleq \overline{v}^{\mathrm{z}}(t).
\label{equ: v_bar_z_ave}
\end{align}
Accordingly, we replace 
$ \overline{v}_n^{\mathrm{z}}(t)$ with $\overline{v}^{\mathrm{z}}(t) $ in (A\ref{alg: gamp_sc}-4, 5, 6).
Furthermore, $\{1/\check{v}_n^{\mathrm{z}}(t)\}_{n=1}^{N}$ in (A\ref{alg: gamp_sc}-6) are also averaged out, leading to
\begin{align}
    \frac{1}{\check{v}^{\mathrm{z}}(t)} 
    &\triangleq
    \frac{1}{N}
    \sum_{n=1}^{N} \frac{1}{\check{v}_n^{\mathrm{z}}(t)}
    \nonumber
    \\
    &
    =
    -
    \underbrace{
    \frac{1}{N}
    \sum_{n=1}^{N} 
    \frac{\partial}{\partial \overline{z}_n}
    g_{\mathrm{out}}
    \left( \overline{z}_n (t); \overline{v}^{\mathrm{z}}(t), y_n \right)}_{ \triangleq \alpha^{\mathrm{out}}(t) },
    \label{equ: v_check_z_approx}
\end{align}
where we define \emph{output divergence} $\alpha^{\mathrm{out}}(t)$, a quantity that represents the average sensitivity of $g_{\mathrm{out}}$ to its input value.
Then, (A\ref{alg: gamp_sc}-7) can be approximated
using the output divergence in \eqref{equ: v_check_z_approx} as
\begin{align}
\overline{v}_m^{\mathrm{x}}{(t)}
&\to
\!
\left( \sum_{n=1}^{N} \frac{|a_{n, m}|^2}{\check{v}^{\mathrm{z}}{(t)}} \right)^{-1}
\!\!\!=
\frac{\check{v}^{\mathrm{z}}(t)}{N}
\biggl(
\underbrace{
\frac{1}{N}\sum_{n=1}^{N} |a_{n, m}|^2
}_{\approx \mathbb{E}\left[ |a_{n, m}|^2 \right] = 1/N}
\biggr)^{-1}
\nonumber
\\
&
\approx
\check{v}^{\mathrm{z}}(t)
=
-\frac{1}{\alpha^{\mathrm{out}}(t)}.
\end{align}
Accordingly, $\overline{v}_m^{\mathrm{x}}(t)$ in (A\ref{alg: gamp_sc}-8, 9) and \eqref{equ: v_hat_x_ave} can be replaced with $-1/\alpha^{\mathrm{out}}(t)$.
Especially, 
\eqref{equ: v_hat_x_ave} reduces to
\begin{align}
&
\hat{v}^{\mathrm{x}}{(t+1)} 
\nonumber
\\
&
\to
-
\frac{1}{\alpha^{\mathrm{out}}(t)} 
\cdot
\underbrace{
\frac{1}{M}
\sum_{m=1}^{M}
\frac{\partial}{\partial \overline{x}_m}
g_{\mathrm{in}}
\left(\overline{x}_m{(t)}; -1/\alpha^{\mathrm{out}}(t) \right)
}_{\triangleq \alpha^{\mathrm{in}}(t+1) }
\nonumber
\\
&=
-
\frac{\alpha^{\mathrm{in}}(t+1)}{\alpha^{\mathrm{out}}(t)},
\label{equ: v_hat_x_approx}
\end{align}
where we define the \emph{input divergence} $ \alpha^{\mathrm{in}}(t) $ similarly to $\alpha^{\mathrm{out}}(t) $.
Finally, substituting \eqref{equ: v_hat_x_approx} into the definition of $ \overline{v}^{\mathrm{z}}(t) $ in \eqref{equ: v_bar_z_ave} yields
\begin{equation}
    \overline{v}^{\mathrm{z}}(t) = 
    \frac{\hat{v}^{\mathrm{x}}(t)}{\xi}
    =
    - \frac{1}{\xi} \cdot \frac{\alpha^{\mathrm{in}}(t)}{\alpha^{\mathrm{out}}(t-1)}.
    \label{equ: v_bar_div_empirical}
\end{equation}

\input{ALG/gamp_sv_div}

The resultant simplified update rules of Algorithm~\ref{alg: gamp_sc} are summarized in Algorithm~\ref{alg: gamp_sv_div}, where we define
\begin{subequations}
\begin{align}
&
\overline{\bm{z}}(t) \triangleq [\overline{z}_1(t), \overline{z}_2(t), \ldots , \overline{z}_N(t)]^{\mathsf{T}} \in \mathbb{C}^{N \times 1},
\\
&
\check{\bm{z}}(t) \triangleq [\check{z}_1(t), \check{z}_2(t), \ldots , \check{z}_N(t)]^{\mathsf{T}}
\in \mathbb{C}^{N \times 1},
\\
&
\overline{\bm{x}}(t) \triangleq [\overline{x}_1(t), \overline{x}_2(t), \ldots , \overline{x}_M(t)]^{\mathsf{T}} \in \mathbb{C}^{M \times 1},
\\
&
\hat{\bm{x}}(t) \triangleq [\hat{x}_1(t), \hat{x}_2(t), \ldots , \hat{x}_M(t)]^{\mathsf{T}}
\in \mathbb{C}^{M \times 1}.
\end{align}
\end{subequations}

Building on the reformulations above, we observe that the analytical tractability of Bayes-\ac{GAMP} critically hinges on whether the decision functions $g_{\mathrm{out}}$ in (A\ref{alg: gamp_sv_div}-4, 5) and $g_{\mathrm{in}}$ in (A\ref{alg: gamp_sv_div}-7, 8) can be derived analytically.
In particular, the output decision function $g_{\mathrm{out}}$ is given by the score function in \eqref{equ: def_g_out},
\textit{i.e.}, the logarithmic derivative of the likelihood
$p_{\mathsf{y}_n \mid \overline{\mathsf{z}}_n}(y_n \mid \overline{z})$,
which arises from the scalar Gaussian-prior model in
Fig.~\ref{fig: scalar_models}(a), rather than from the standard \ac{AWGN} model in Fig.~\ref{fig: scalar_models}(b).
Furthermore, this likelihood can be expressed as the convolution of the observation likelihood $p_{\mathsf{y}_n \mid \mathsf{z}_n}(y_n \mid z)$ and a noncentral Gaussian \ac{PDF} with respect to $z$, as shown in \eqref{equ: p_y_z_ol}.
Therefore, the analytical derivation of $g_{\mathrm{out}}$ requires precisely the two conditions stated in the Introduction as (C1) and (C2) in Section~\ref{chap: intro}:
the availability of an explicit expression for the observation likelihood associated with the nonlinear mapping $h$, and the closed-form evaluability of the corresponding convolution integral.
%
Satisfying both conditions is generally challenging, especially for complex-valued Bayes-\ac{GAMP}.
Although (C2) can, in principle, be addressed through numerical integration, its evaluation essentially reduces to a computationally demanding two-dimensional real-valued integral for each output component, rendering its direct use almost prohibitive within the iterative complex-valued \ac{GAMP} framework.

Motivated by this observation, the next section develops a principled methodology for training an \ac{NN} that emulates the behavior of $g_{\mathrm{out}}$ within the Bayes-\ac{GAMP} recursion without requiring explicit evaluation of the underlying score function.
%

%% file: ALG/gamp.tex


\begin{algorithm}[!t]
\caption{Bayes-GAMP (via Denoiser)~\cite{Rangan2011, Schniter2015}}
\label{alg: gamp}
\begin{algorithmic}[1]
\Require{$\bm{y} \in \mathbb{C}^{N \times 1}, \bm{A} \in \mathbb{C}^{N \times M}$.}
\Ensure{$ \{ \hat{x}_m{(T+1)}\}_{m=1}^{M}$.}


\State $\forall m : \hat{x}_m (1) = 0, \hat{v}_m^{\mathrm{x}} (1) = \sigma_{\mathrm{x}}^2, \, 
\forall n : \check{z}_n (0) = 0$
%
%
\For{$t=1 \,\, \textrm{to} \,\, T$}

\!\!\!\!\!\!\!\!\!\!\!\! /* Output Module */

\State $\forall n : \overline{v}^{\mathrm{z}}_n(t) = \sum_{m=1}^{M} |a_{n, m}|^2 \hat{v}_m^{\mathrm{x}}(t) $
\State $\forall n : \overline{z}_n(t) = \sum_{m=1}^{M} a_{n, m} \hat{x}_m(t) - \overline{v}^{\mathrm{z}}_n(t) \check{z}_n(t-1) $
\State $\forall n : \hat{z}_n(t) = \eta_{\mathrm{out}}
\left( \overline{z}_n (t); \overline{v}_n^{\mathrm{z}}(t), y_n \right)$
\State $\forall n : \hat{v}_n^{\mathrm{z}}(t) =
\overline{v}_n^{\mathrm{z}}(t)
\cdot
\eta_{\mathrm{out}}'
\left( \overline{z}_n (t); \overline{v}_n^{\mathrm{z}}(t), y_n \right)$
\State $\forall n : \check{z}_n(t) = \frac{\hat{z}_n(t) - \overline{z}_n(t)}{\overline{v}^{\mathrm{z}}_n(t)} 
$
\State $\forall n : \check{v}_n^{\mathrm{z}}(t) =
\frac{\overline{v}^{\mathrm{z}}_n(t)}{1 - \hat{v}_n^{\mathrm{z}}(t) / \overline{v}^{\mathrm{z}}_n(t)}$

\!\!\!\!\!\!\!\!\!\!\!\!  /* Input Module */

\State $ \forall m : \overline{v}_m^{\mathrm{x}}(t) = \left( \sum_{n=1}^{N} \frac{|a_{n, m}|^2}{\check{v}_n^{\mathrm{z}} (t)}   \right)^{-1} $
\State $ \forall m : \overline{x}_m (t) = \hat{x}_m (t) + \overline{v}_m^{\mathrm{x}}(t) \sum_{n=1}^{N} a_{n, m}^* \check{z}_n(t) $
\State $\forall m : \hat{x}_m(t+1) =  \eta_{\mathrm{in}} \left(\overline{x}_m{(t)}; \overline{v}_m^{\mathrm{x}}{(t)} \right)$
\State $\forall m : \hat{v}_m^{\mathrm{x}}{(t+1)} = 
\overline{v}_m^{\mathrm{x}}{(t)}
\cdot
\eta_{\mathrm{in}}' 
\left(\overline{x}_m{(t)}; \overline{v}_m^{\mathrm{x}}{(t)} \right)$
\EndFor
\end{algorithmic}
\end{algorithm}

%% file: ALG/gamp_sc.tex


\begin{algorithm}[!t]
\caption{Bayes-GAMP (via Score)~\cite{Rangan2011, Schniter2015}}
\label{alg: gamp_sc}
\begin{algorithmic}[1]
\Require{$\bm{y} \in \mathbb{C}^{N \times 1}, \bm{A} \in \mathbb{C}^{N \times M}$.}
\Ensure{$ \{ \hat{x}_m{(T+1)}\}_{m=1}^{M}$.}


\State $\forall m : \hat{x}_m (1) = 0, \hat{v}_m^{\mathrm{x}} (1) = \sigma_{\mathrm{x}}^2, \, 
\forall n : \check{z}_n (0) = 0$
%
%
\For{$t=1 \,\, \textrm{to} \,\, T$}

\!\!\!\!\!\!\!\!\!\!\!\! /* Output Module */

\State $\forall n : \overline{v}^{\mathrm{z}}_n(t) = \sum_{m=1}^{M} |a_{n, m}|^2 \hat{v}_m^{\mathrm{x}}(t) $
\State $\forall n : \overline{z}_n(t) = \sum_{m=1}^{M} a_{n, m} \hat{x}_m(t) - \overline{v}^{\mathrm{z}}_n(t) \check{z}_n(t-1) $
\State $\forall n : \check{z}_n(t) = g_{\mathrm{out}}
\left( \overline{z}_n (t); \overline{v}_n^{\mathrm{z}}(t), y_n \right)$
\State $\forall n : \frac{1}{\check{v}_n^{\mathrm{z}}(t)} =
-
g_{\mathrm{out}}'
\left( \overline{z}_n (t); \overline{v}_n^{\mathrm{z}}(t), y_n \right)$

\!\!\!\!\!\!\!\!\!\!\!\!  /* Input Module */

\State $ \forall m : \overline{v}_m^{\mathrm{x}}(t) = \left( \sum_{n=1}^{N} \frac{|a_{n, m}|^2}{\check{v}_n^{\mathrm{z}} (t)}   \right)^{-1} $
\State $ \forall m : \overline{x}_m (t) = \hat{x}_m (t) + \overline{v}_m^{\mathrm{x}}(t) \sum_{n=1}^{N} a_{n, m}^* \check{z}_n(t) $
\State $\forall m : \hat{x}_m(t+1) =  g_{\mathrm{in}} \left(\overline{x}_m{(t)}; \overline{v}_m^{\mathrm{x}}{(t)} \right)$
\State $\forall m : \hat{v}_m^{\mathrm{x}}{(t+1)} = 
\overline{v}_m^{\mathrm{x}}{(t)}
\cdot
g_{\mathrm{in}}' \left(\overline{x}_m{(t)}; \overline{v}_m^{\mathrm{x}}{(t)} \right)$
\EndFor
\end{algorithmic}
\end{algorithm}

%% file: ALG/gamp_sv_div.tex


\begin{algorithm}[!t]
\caption{Bayes-GAMP (via Divergence)}
\label{alg: gamp_sv_div}
\begin{algorithmic}[1]
\Require{$\bm{y} \in \mathbb{C}^{N \times 1}, \bm{A} \in \mathbb{C}^{N \times M}$.}
\Ensure{$ \bm{\hat{x}}{(T+1)}$.}


\State $\hat{\bm{x}}(1) = \bm{0}, \, \check{\bm{z}}(0) = \bm{0}, \, - \frac{\alpha^{\mathrm{in}}(1)}{\alpha^{\mathrm{out}}(0)} = \sigma_{\mathrm{x}}^2.$
%
%
\For{$t=1 \,\, \textrm{to} \,\, T$}

\!\!\!\!\!\!\!\!\!\!\!\! /* Output Module */

\State $\overline{\bm{z}}(t)
=
\bm{A} \hat{\bm{x}}(t)
+
\frac{1}{\xi}
\cdot
\frac{\alpha^{\mathrm{in}}(t)}{\alpha^{\mathrm{out}}(t-1)}
\check{\bm{z}}(t-1) $
\State $\check{\bm{z}}(t)
=
g_{\mathrm{out}}
\left(
\overline{\bm{z}}(t);
-
\frac{1}{\xi}
\cdot
\frac{\alpha^{\mathrm{in}}(t)}{\alpha^{\mathrm{out}}(t-1)},
\bm{y}
 \right) $
\State $\alpha^{\mathrm{out}} (t)
=
\left\langle
g_{\mathrm{out}}'
\left(
\overline{\bm{z}}(t);
-
\frac{1}{\xi}
\cdot
\frac{\alpha^{\mathrm{in}}(t)}{\alpha^{\mathrm{out}}(t-1)},
\bm{y}
 \right)
\right\rangle$

\!\!\!\!\!\!\!\!\!\!\!\!  /* Input Module */

\State $\overline{\bm{x}}(t)
=
\hat{\bm{x}}(t)
-
\frac{1}{\alpha^{\mathrm{out}}(t)}
\bm{A}^{\mathsf{H}}
\check{\bm{z}}(t)$
\State $ \hat{\bm{x}}(t+1)
=
g_{\mathrm{in}}
\left(
\overline{\bm{x}}(t);
-\frac{1}{\alpha^{\mathrm{out}}(t)}
 \right) $
\State $\alpha^{\mathrm{in}} (t+1)
=
\left\langle
g_{\mathrm{in}}'
\left(
\overline{\bm{x}}(t);
-
\frac{1}{\alpha^{\mathrm{out}}(t)}
 \right)
\right\rangle$
\EndFor
\end{algorithmic}
\end{algorithm}

%% file: TXT/se_dsm.tex
\subsection{Motivation and Challenges}

In this section, we consider training a network
$s_{\varTheta}(\overline{z}, y; \overline{v})$
with learnable parameters $\varTheta$ such that 
$s_{\varTheta}(\overline{z}, y; \overline{v}) 
\approx g_{\mathrm{out}}(\overline{z}; \overline{v}, y)$,
where $g_{\mathrm{out}}$ denotes the output decision function defined in \eqref{equ: def_g_out}.
A straightforward approach would be to train
$s_\varTheta : (\mathbb{C} \times \mathbb{C} \times \mathbb{R}_+) \to \mathbb{C}$
by minimizing the following \ac{MSE}:
\begin{align}
    J(\varTheta; \overline{v})
    &\triangleq
    \mathbb{E}_{\overline{\mathsf{z}}, \mathsf{y}}
    \left[ 
    \left| 
    s_{\varTheta} (\overline{z}, y; \overline{v}) - g_{\mathrm{out}} (\overline{z}; \overline{v}, y)
    \right|^2
    \right]
    \nonumber
    \\
    &
    \!\!\!\!\!\!\!\!\!\!\!\!\!\!
    \!\!\!\!\!\!\!
    =
    \mathbb{E}_{\overline{\mathsf{z}}, \mathsf{y}}
    \left[ 
    \left| 
    s_{\varTheta} (\overline{z}, y; \overline{v}) - \frac{\partial}{\partial \overline{z}^*}
\log{
p_{\mathsf{y}_n \mid \overline{\mathsf{z}}_n }
(y \mid \overline{z})}
    \right|^2
    \right].
\label{equ: naive_cost_func}
\end{align}
However, implementing this training procedure in practice involves two major challenges:
\begin{itemize}
\item How can training samples $(\overline{z}, y)$ be generated for evaluating \eqref{equ: naive_cost_func}?

\item How can the loss in \eqref{equ: naive_cost_func} be evaluated when the true score function $g_{\mathrm{out}}$ is not available in closed form?
\end{itemize}
To elaborate on the first issue, during the execution phase of Algorithm~\ref{alg: gamp_sv_div}, $\overline{z} = \overline{z}_n$ is a deterministic quantity given by  (A\ref{alg: gamp_sv_div}-3). 
Consequently, the mean of the Gaussian prior in Fig.~\ref{fig: scalar_models}(a) is fixed, and samples $(z,y)\sim(z_n,y_n)$ conditioned on $\overline{z}=\overline{z}_n$ can be easily generated according to the Gaussian-prior model in Fig.~\ref{fig: scalar_models}(a).
In contrast, training based on \eqref{equ: naive_cost_func}
requires evaluating the expectation with respect to the empirical
distribution of $(\overline{z}, y)$.
Hence, one must jointly generate sample pairs
$(\overline{z},y)$ \emph{from scratch} according to this empirical distribution.
However, characterizing this distribution is difficult when
$g_{\mathrm{out}}$ is unavailable in closed form.
%
The second issue is also nontrivial.
Even if suitable training samples $(\overline{z},y)$ could be generated,
direct minimization of \eqref{equ: naive_cost_func} remains intractable
without a closed-form expression for $g_{\mathrm{out}}$.

\subsection{Main Claim}

The following proposition provides a principled means of circumventing the aforementioned challenges.

\begin{proposition}[SE-DSM]
\label{prop: se_dsm}
Let $\left(\overline{Z}_t, Z, Y\right)$ be random variables whose joint \ac{PDF}\footnote{
For simplicity, this paper assumes that $Y$ is continuous-valued and that
$ p_{\mathsf{Y} \mid \mathsf{Z}} (y\,|\,z) $
is a density with respect to the two-dimensional Lebesgue measure on $\mathbb{C}$.
However, the discussion remains valid for discrete-valued and mix cases by replacing the Lebesgue measure with an appropriate reference measure $\mu$ (\textit{e.g.}, the counting measure), interpreting
$ p_{\mathsf{Y} \mid \mathsf{Z}} (y \,|\, z) $
as the corresponding Radon-Nikodym derivative with respect to $\mu$, and replacing
$\mathrm{d}y$
with
$\mathrm{d}\mu(y)$
in all integrals. 
See Section~2 of \cite{Feng2022} for details.
} is given by
\begin{align}
&
p_{\overline{\mathsf{Z}}_t, \mathsf{Z}, \mathsf{Y}} (\overline{z}, z, y)
=
p_{ \overline{\mathsf{Z}}_t } (\overline{z})
\cdot
p_{ \mathsf{Z} \mid \overline{\mathsf{Z}}_t }
( z \mid \overline{z} )
\cdot
p_{ \mathsf{Y} \mid \mathsf{Z} }
( y \mid z )
\nonumber
\\
&
\!\!\!\!
=
\frac{e^{ - \frac{ |\overline{z}|^2 }{  \xi^{-1} \sigma_{\mathrm{x}}^2 - \overline{v} } }}{\pi \! \left( \xi^{-1} \sigma_{\mathrm{x}}^2 - \overline{v} \right)}
\cdot
\frac{e^{- \frac{| z - \overline{z} |^2}{\overline{v}} }}{\pi \overline{v}}
\cdot
p_{ \mathsf{Y} \mid \mathsf{Z} }
( y \mid z ).
\label{equ: se_joint_pdf}
\end{align}
Using these random variables, define the loss function
\begin{equation}
    J_{\mathrm{DSM}}(\varTheta; \overline{v})
    \triangleq
    \mathbb{E}_{ \overline{\mathsf{Z}}_t, \mathsf{Z}, \mathsf{Y} }
    \!
    \left[ 
    \left| 
    s_{\varTheta} \left(\overline{Z}_t, Y; \overline{v}\right)
    +
    \frac{ \overline{Z}_t - Z }{\overline{v}}
    \right|^2
    \right],
    \label{equ: dsm}
\end{equation}
and its minimizer
\begin{equation}
    \varTheta^\star \triangleq \underset{\varTheta}{\mathrm{arg\,min}} \, J_{\mathrm{DSM}} (\varTheta; \overline{v}).
    \label{equ: def_theta_star}
\end{equation}
Assume that both output decision function
$g_{\mathrm{out}}\left(\overline{z}; \overline{v}, h(z, w) \right)$
and the network $s_{\varTheta} (\overline{z}, h(z, w); \overline{v})$
are Lipschitz continuous with respect to $(\overline{z},z,w)$ and that the assumptions stated in Section~\ref{chap: signal_model_assumption} hold\footnote{
The Bayes-optimal input denoiser $g_{\mathrm{in}} (\overline{x}; \overline{v}) = \eta_{\mathrm{in}} (\overline{x}; \overline{v})$ is Lipschitz continuous; see \cite[Lemma~2]{Takeuchi2020}.
}.
Also,  $s_{\varTheta}$ is assumed to be sufficiently expressive to contain the true score function in \eqref{equ: def_g_out}.
Then, under certain regularity conditions\footnote{
Specifically,
$\mathbb{E}_{\overline{\mathsf{Z}}_t, \mathsf{Y}}
\!
\left[ 
\left| 
s_{\varTheta} (\overline{Z}_t, Y; \overline{v}) 
\right|^2
\right], \,
\mathbb{E}_{\overline{\mathsf{Z}}_t, \mathsf{Y}}
\!
\left[ 
\left| 
g_{\mathrm{out}} (\overline{Z}_t; \overline{v}, Y)
\right|^2
\right] < \infty
$
and $p_{\mathsf{Y} \mid \mathsf{Z}} (y \,|\, z)$ admits interchange of differentiation and integration with respect to $z$. 
A sufficient condition is that $p_{\mathsf{Y} \mid \mathsf{Z}} (y\,|\, z)$ grows at most polynomially in $z$.
}, 
Algorithm~\ref{alg: gamp_sv_div} and the modified \ac{GAMP} algorithm obtained by replacing
$g_{\mathrm{out}}(\overline{z};\overline{v},y)$
with
$s_{\varTheta^\star}(\overline{z},y;\overline{v})$
in steps (A\ref{alg: gamp_sv_div}-4) and (A\ref{alg: gamp_sv_div}-5) achieve, almost surely in the large-system limit, the same asymptotic \ac{MSE}, 
\begin{equation}
\label{equ: def_mse}
 \mathrm{MSE}(t) \triangleq 
 \lim_{N = \xi M \to \infty}
 \frac{1}{M} \left\| \hat{\bm{x}}(t+1) - \bm{x} \right\|^2,
\end{equation}
at each iterative step
$ t = 1, 2, \ldots , T \, (< \infty)$.
\end{proposition}

Proposition~\ref{prop: se_dsm} implies that, if sufficiently rich $s_{\varTheta}$ is trained according to Algorithm~\ref{alg: training_nn}
and converges to a globally optimal solution, then replacing
$g_{\mathrm{out}}$
in Algorithm~\ref{alg: gamp_sv_div} with
$s_{\varTheta}$
does not alter the asymptotically achievable performance almost surely.
The key observation here is that training samples can be generated by drawing
$\overline{z}$ in \eqref{equ: naive_cost_func} from the Gaussian distribution specified in (A\ref{alg: training_nn}-2), instead of the empirical distribution appearing in \eqref{equ: naive_cost_func}.
Furthermore, the network is trained to predict the artificially injected Gaussian noise
$\bm{b}'$ introduced in (A\ref{alg: training_nn}-4), not the intractable score function.
We refer to this training methodology as \emph{SE-based DSM (SE-DSM)}.
At first glance, this result may seem a straightforward adaptation of the classical \ac{DSM} result~\cite{Vincent2011} to the \ac{GAMP} setting.
However, \ac{DSM} typically assumes the \ac{AWGN} model in Fig.~\ref{fig: scalar_models}(b), while SE-DSM is built upon the Gaussian-prior model in Fig.~\ref{fig: scalar_models}(a); 
hence, this proposition is far from trivial.

A particularly noteworthy aspect of Algorithm~\ref{alg: training_nn} is that it can be implemented as long as the forward evaluation of the nonlinear mapping $h$ is available.
In particular, the training procedure requires none of an explicit analytical expression of $h$, the derivation of the corresponding likelihood function
$p_{\mathsf{y}_n \mid \mathsf{z}_n} (y \mid z)$,
or knowledge of the output function (\textit{i.e.}, the resultant score function)
$g_{\mathrm{out}} (\overline{z}; \overline{v}, y)$.
Thus, Proposition~\ref{prop: se_dsm} suggests a remarkable possibility: 
for an extremely broad class of nonlinear mappings $h$, possibly including black-box ones, Bayes-\ac{GAMP} can be realized solely through a network
$s_{\varTheta}$ with sufficient representational capacity.
Furthermore, 
the fixed point of \ac{SE} equations of Algorithm~\ref{alg: gamp_sv_div} has been rigorously proven to coincide with the Bayes-optimal one,  provided that it is unique~\cite{Barbier2019, Takeuchi2025}.
Combining this fact with Proposition~\ref{prop: se_dsm} further suggests the asymptotic Bayes-optimality of Algorithm~\ref{alg: gamp_sv_div} with the globally optimal network $s_{\varTheta^\star}$ under suitable technical conditions.

We briefly comment on Proposition~\ref{prop: se_dsm}.
The rigor of its proof ultimately relies on the validity of the \ac{SE} characterization of \ac{GAMP}.
Several works have established rigorous \ac{SE} analyses for \ac{GAMP}, including \cite{Rangan2012arxiv, Javanmard2013, Feng2022};
however, these results are restricted to the real-valued setting.
\Ac{GAMP} itself was explicitly extended to the complex domain in \cite{Schniter2015, Zou2018}, and rigorous \ac{SE}-based analyses in the complex domain have subsequently appeared in \cite{Ma2019TIT, Mondelli2022}.
Nevertheless, these works do not constitute a direct complex-domain extension of the rigorous results in \cite{Rangan2012arxiv, Javanmard2013, Feng2022}.
Accordingly, the claims made in this paper should be understood as holding under the assumption that the \ac{SE} characterization of \ac{GAMP} remains valid in the complex domain.
That said, rigorous complex-domain \ac{SE} analyses have been established for Bayes-\ac{OAMP}/\ac{VAMP}~\cite{Ma2017, Rangan2019}; see, \textit{e.g.}, \cite{Takeuchi2020}.
In light of these results, it is natural to conjecture that an analogous rigorous \ac{SE} characterization also holds for \ac{GAMP} in the complex domain.


\input{ALG/train_nn}

%% file: ALG/train_nn.tex


\begin{algorithm}[!t]
\caption{Training Procedure based on SE-DSM}
\label{alg: training_nn}
\begin{algorithmic}[1]
\Require{$ \overline{v} \in \left( 0, \xi^{-1}  \sigma_{\mathrm{x}}^2 \right)$,  Batch Size $N_{\mathrm{B}}$, Learning Rate $\eta$.}


\Repeat

\State Sample $\overline{\bm{z}}' \sim \mathcal{CN} \left( \bm{0}, \left( \xi^{-1}\sigma_{\mathrm{x}}^2 - \overline{v} \right) \bm{I}_{N_{\mathrm{B}}} \right)$.

\State Sample $ \bm{b}' \sim \mathcal{CN} \left( \bm{0}, \overline{v}  \,\bm{I}_{N_{\mathrm{B}}} \right) $  independently of $\overline{\bm{z}}'$.

\State Calculate $ \bm{z}' = \overline{\bm{z}}' +  \bm{b}'$.

\State Sample $ \{w_i'\}_{i=1}^{{N_{\mathrm{B}}}} \overset{\mathrm{i.i.d.}}{\sim} W$ independently of $(\overline{\bm{z}}', \bm{b}')$ and define $ \bm{w}' \triangleq [w_1', w_2', \ldots , w_{N_{\mathrm{B}}}']^{\mathsf{T}} \in \mathbb{C}^{{N_{\mathrm{B}}} \times 1}$.

\State Calculate $ \bm{y}' = h(\bm{z}', \bm{w}')$.

\State Update $ \varTheta \leftarrow \varTheta - \eta \cdot \bm{\nabla}_{\varTheta} \left\| s_{\varTheta} (\overline{\bm{z}}', \bm{y}'; \overline{v}) + \frac{\overline{\bm{z}}' - \bm{z}'}{\overline{v}} \right\|^2 $

\Until{converged}
\end{algorithmic}
\end{algorithm}

%% file: TXT/proof.tex
\input{ALG/gamp_se}

In this section, we sketch the proof of Proposition~\ref{prop: se_dsm} under the assumption that the \ac{SE} equations of \ac{GAMP} remain valid in the complex domain.
The proof proceeds in three steps:
\begin{enumerate}
    \item[(i)] Based on the \ac{SE} characterization, we derive the joint distribution in \eqref{equ: se_joint_pdf}, which is used to generate the training data.
    \item[(ii)] We establish the equivalence, up to an additive constant independent of $\varTheta$, between the explicit score-matching objective induced by this distribution and the tractable loss function $J_{\mathrm{DSM}}$ in \eqref{equ: dsm}.
    \item[(iii)] We verify that replacing $g_{\mathrm{out}}$ with $s_{\varTheta^\star}$ does not alter the asymptotically achievable \ac{MSE}.
\end{enumerate}
The following subsections detail these steps.

\subsection{Derivation of \eqref{equ: se_joint_pdf} based on SE Characterization}

According to the \ac{SE} analysis of Bayes-GAMP~\cite{Rangan2012arxiv, Javanmard2013, Feng2022}, under the assumptions of Proposition~\ref{prop: se_dsm}, the signals in \eqref{equ: glm} and the messages propagating in Algorithm~\ref{alg: gamp_sv_div} empirically converge, in the large-system limit $N=\xi M \to \infty$, to a certain collection of random variables:
\begin{subequations}
\begin{align}
&
\label{equ: se_conv_out}
\bigl\{z_n, \overline{z}_n{(t)}, w_n, y_n \bigr\}_{n=1}^{N} \overset{\mathrm{PL}(\ell)}{\longrightarrow}
\left(Z, \overline{Z}_t, W, Y \right),
\\
&
\label{equ: se_conv_in}
\bigl\{x_m, \overline{x}_m{(t)} \bigr\}_{m = 1}^{M} \overset{\mathrm{PL}(\ell)}{\longrightarrow}
\left( X, \overline{X}_t \right),
\end{align}
\end{subequations}
where the random variables
$X, \big\{\overline{X}_t\bigr\}_{t=1}^{T}, Z, $ and  $\bigl\{\overline{Z}_t\bigr\}_{t=1}^{T}$
are characterized by Algorithm~\ref{alg: gamp_se},
with Gaussian random variables $\{H_t\}_{t=1}^{T}$ in (A\ref{alg: gamp_se}-8) independent across $t = 1, 2, \ldots , T$ and of all other random variables.
Moreover, the divergences in (A\ref{alg: gamp_sv_div}-5, 8) converge, almost surely in the large-system limit $N=\xi M \to \infty$, to deterministic quantities recursively defined by (A\ref{alg: gamp_se}-7, 9)~[\citenum{Rangan2012arxiv}, Claim~1], [\citenum{Feng2022}, Remark~4.3]. 
Specifically, for $t = 1, 2, \ldots, T$, 
\begin{equation}
    \alpha^{\mathrm{out}} (t) \overset{\mathrm{a.s.}}{\to} \overline{\alpha}_t^{\mathrm{out}}, \quad \alpha^{\mathrm{in}} (t+1) \overset{\mathrm{a.s.}}{\to} \overline{\alpha}_{t+1}^{\mathrm{in}}
    \label{equ: div_conv}
\end{equation}
holds.
Furthermore, the random variable $W$ in \eqref{equ: se_conv_out} and (A\ref{alg: gamp_se}-7) is also independent of all other random variables.
Finally, the random variable $Y$ in \eqref{equ: se_conv_out} is defined by $Y = h(Z,W)$.

Based on the above existing results, for all $t = 2, 3, \ldots , T$, 
the predicted variance $\overline{v} = \overline{v}^{\mathrm{z}} (t)$, which is fed with the output decision function $g_{\mathrm{out}} (\overline{z}; \overline{v}, y)$, converges almost surely to the following value in the large-system limit:
\begin{align}
    &
    \overline{v} 
    = \overline{v}^{\mathrm{z}} (t)
    \overset{\textrm{(a)}}{=}
    - \frac{1}{\xi} \cdot \frac{\alpha^{\mathrm{in}}(t)}{\alpha^{\mathrm{out}}(t-1)}
    \nonumber
    \\
    &
    \overset{\textrm{(b)}}{\overset{\mathrm{a.s.}}{\to}}
    -\frac{1}{\xi} \cdot \frac{\overline{\alpha}_t^{\mathrm{in}}}{\overline{\alpha}_{t-1}^{\mathrm{out}}}
    \overset{\textrm{(c)}}{=}
    \frac{1}{\xi}
    \cdot
    \!
    \left( \! - \frac{1}{\overline{\alpha}_{t-1}^{\mathrm{out}}} \! \right) \!
    \cdot
    \mathbb{E}_{\overline{\mathsf{X}}_t}
    \left[
    g_{\mathrm{in}}'
    \!
    \left(
    \!
    \overline{X}_t;
    -
    \frac{1}{\overline{\alpha}^{\mathrm{out}}_{t-1}}
    \!
    \right)
    \!
    \right]
    \nonumber
    \\
    &
    \overset{\textrm{(d)}}{=}
    \frac{1}{\xi}
    \cdot
    \mathbb{E}_{\mathsf{X}, \overline{\mathsf{X}}_{t-1}}
    \left[ 
    \left| X - g_{\mathrm{in}} \left( \overline{X}_{t-1}; - \frac{1}{\overline{\alpha}_{t-1}^{\mathrm{out}}} \right) \right|^2
    \right]
    \nonumber
    \\
    &=
    \frac{1}{\xi}
    \cdot
    \Biggl(
    \underbrace{
    \mathbb{E}_{\mathsf{X}}
    \left[ |X|^2 \right]}_{= \sigma_{\mathrm{x}}^2}
    -
    \,
    \mathbb{E}_{\mathsf{X}, \overline{\mathsf{X}}_{t-1}}
    \left[ X^* g_{\mathrm{in}}
    \left( \overline{X}_{t-1}; - \frac{1}{\overline{\alpha}_{t-1}^{\mathrm{out}}} \right) \right]
    \nonumber
    \\
    &
    \qquad \qquad
    -
    \,
    \mathbb{E}_{\mathsf{X}, \overline{\mathsf{X}}_{t-1}}
    \left[ X \, g_{\mathrm{in}}
    \left( \overline{X}_{t-1}; - \frac{1}{\overline{\alpha}_{t-1}^{\mathrm{out}}} \right)^* \right]
    \nonumber
    \\
    &
    \qquad \qquad \qquad \qquad
    +
    \,
    \mathbb{E}_{\overline{\mathsf{X}}_{t-1}}
    \left[
    \left|
    g_{\mathrm{in}}
    \biggl( \overline{X}_{t-1}; - \frac{1}{\overline{\alpha}_{t-1}^{\mathrm{out}}} \biggr)
    \right|^2
    \right]
    \Biggr)
    \nonumber
    \\
    &
    \overset{\textrm{(e)}}{=}
    \frac{1}{\xi}
    \cdot
    \left(
    \sigma_{\mathrm{x}}^2
    -
    \mathbb{E}_{\overline{\mathsf{X}}_{t-1}}
    \left[
    \left|
    g_{\mathrm{in}}
    \biggl( \overline{X}_{t-1}; - \frac{1}{\overline{\alpha}_{t-1}^{\mathrm{out}}} \biggr)
    \right|^2
    \right]
    \right),
    \label{equ: v_bar_g_in_rep}
\end{align}
where steps (a) and (b) respectively follow from \eqref{equ: v_bar_div_empirical} and \eqref{equ: div_conv};
(c) from the definition of the asymptotic divergence in (A\ref{alg: gamp_se}-9);
(d) from the well-known identity of input denoiser [\citenum{Takeuchi2020}, Lemma~2], along with the Gaussianity of $H_t$ in (A\ref{alg: gamp_se}-8);
(e) from the following fundamental property of posterior mean~\cite{Kay1993}, \textit{i.e.}, 
\begin{equation}
\mathbb{E}_{\mathsf{X}, \overline{\mathsf{X}}_t} 
\left[ 
\left( X - \mathbb{E}_{\mathsf{X}} \left[X \mid \overline{X}_t \right] \right)^*
\mathbb{E}_{\mathsf{X}} \left[X \mid \overline{X}_t \right]
\right]
=
0,
\end{equation}
as well as the definition of input decision function in \eqref{equ: def_g_in}, namely, 
$g_{\mathrm{in}} \left(\overline{X}_{t}; -1/\overline{\alpha}_t^{\mathrm{out}} \right) = \mathbb{E}_{\mathsf{X}} \left[ X \mid \overline{X}_t \right] $.

Meanwhile, according to the \ac{SE} characterization, the signal-message pairs
$ \bigl\{z_n, \overline{z}_n{(t)} \bigr\}_{n=1}^{N} $
empirically converge to jointly Gaussian random variables characterized by (A\ref{alg: gamp_se}-5).
Specifically, (A\ref{alg: gamp_se}-5) can be rewritten in terms of $\overline{v}$ by substituting
\begin{equation}
\xi^{-1}
\cdot
    \mathbb{E}_{\overline{\mathsf{X}}_{t-1}}
    \!
    \left[
    \left|
    g_{\mathrm{in}}
    \biggl( \overline{X}_{t-1}; - \frac{1}{\overline{\alpha}_{t-1}^{\mathrm{out}}} \biggr)
    \right|^2
    \right]
    \!
    =
    \xi^{-1} \cdot \sigma_{\mathrm{x}}^2 - \overline{v},
\end{equation}
which follows from \eqref{equ: v_bar_g_in_rep}, into (A\ref{alg: gamp_se}-5), arriving at
\begin{equation}
\begin{bmatrix}
Z \\
\overline{Z}_t \\
\end{bmatrix}
\sim
\mathcal{CN}
\left( 
\begin{bmatrix}
0 \\
0 \\
\end{bmatrix},
\begin{bmatrix}
\xi^{-1} 
\sigma_{\mathrm{x}}^2 & 
\xi^{-1}
\sigma_{\mathrm{x}}^2
-
\overline{v}
\\
\xi^{-1}
\sigma_{\mathrm{x}}^2
-
\overline{v}
&
\xi^{-1}
\sigma_{\mathrm{x}}^2
-
\overline{v} \\
\end{bmatrix}
\right).
\label{equ: Z_cov}
\end{equation}
A convenient realization of $ \left(Z,\overline{Z}_t \right)$ satisfying the above covariance structure is to define independent Gaussian random variables
$\overline{Z}_t \sim \mathcal{CN} \left(0, \xi^{-1}\sigma_{\mathrm{x}}^2 - \overline{v} \right) $ and
$B_t \sim \mathcal{CN} \left(0, \overline{v} \right)$, obtaining $Z$ as $ Z = \overline{Z}_t + B_t $, since such a construction satisfies
\begin{subequations}
\begin{align}
&
\mathbb{E}_{\mathsf{Z}} \left[ |Z|^2 \right]
=
\mathbb{E}_{\overline{\mathsf{Z}}_{t}} \left[ \bigl|\overline{Z}_t \bigr|^2 \right] + \mathbb{E}_{\mathsf{B}_t} \left[ \bigl| B_t \bigr|^2 \right] = \xi^{-1} \sigma_{\mathrm{x}}^2,
\\
&
\mathbb{E}_{\mathsf{Z}, \overline{\mathsf{Z}}_t} \left[Z \overline{Z}_t^* \right]
=
\mathbb{E}_{\overline{\mathsf{Z}}_{t}} \left[ \bigl|\overline{Z}_t \bigr|^2 \right]
+
\underbrace{
\mathbb{E}_{\overline{\mathsf{Z}}_t} \left[ \overline{Z}_t^* \right]
\cdot
\mathbb{E}_{\mathsf{B}_t} [ B_t ]}_{= 0 \cdot 0 = 0}
\nonumber
\\
&\qquad \qquad \quad \,\,
=
\xi^{-1} \sigma_{\mathrm{x}}^2 - \overline{v}.
\end{align}
\end{subequations}
Consequently, the triplet
$\left(\overline{Z}_t, Z, Y\right)$
constructed in this manner follows the joint distribution with density given by \eqref{equ: se_joint_pdf}.

\subsection{Reformulation of the Loss Function}

The preceding discussion shows that, in the large-system limit, the signal-message tuples $ \bigl\{\overline{z}_n{(t)},  z_n, y_n \bigr\}_{n=1}^{N} $ empirically converge to a random triplet 
$ \left(\overline{Z}_t, Z, Y \right)$ following a distribution consistent with the asymptotic input $\overline{v}$ to $g_{\mathrm{out}}$, \textit{i.e.}, the one with the joint \ac{PDF} in \eqref{equ: se_joint_pdf}.
Therefore, assuming the large-system limit, one can safely replace the loss function in \eqref{equ: naive_cost_func} with
\begin{align}
&
J_{\mathrm{SM}}(\varTheta; \overline{v})
\nonumber
\\
&
\!\!\!\!\!
\triangleq
\mathbb{E}_{\overline{\mathsf{Z}}_t, \mathsf{Y}}
\!
\left[ 
\left| 
s_{\varTheta} (\overline{Z}_t, Y; \overline{v}) - \frac{\partial}{\partial \overline{Z}_t^*}
\log{
p_{\mathsf{Y} \mid \overline{\mathsf{Z}}_t }
\left(Y \mid \overline{Z}_t \right)}
\right|^2
\right].
\!
\label{equ: e_dsm}
\end{align}
In the following, we establish the equivalence between \eqref{equ: dsm} and \eqref{equ: e_dsm}, up to an additive constant.
To this end, expanding \eqref{equ: e_dsm}, we obtain
\begin{align}
&J_{\mathrm{SM}}(\varTheta; \overline{v})
=
\underbrace{
\mathbb{E}_{ \overline{\mathsf{Z}}_t, \mathsf{Y} }
\left[ 
\left|
s_\varTheta \left(\overline{Z}_t, Y; \overline{v} \right) 
\right|^2
\right]}_{= \textrm{(a)} }
\nonumber
\\
&
\quad
- \,
\underbrace{
\mathbb{E}_{ \overline{\mathsf{Z}}_t, \mathsf{Y} }
\left[ 
s_\varTheta \left(\overline{Z}_t, Y; \overline{v} \right)^*
\cdot
\frac{\partial}{\partial \overline{Z}_t^*}
\log{
p_{\mathsf{Y} \mid \overline{\mathsf{Z}}_t }
\left(Y \mid \overline{Z}_t \right)}
\right]}_{= \textrm{(b)}}
\nonumber
\\
&
\quad
- \,
\left(\mathbb{E}_{ \overline{\mathsf{Z}}_t, \mathsf{Y} }
\left[ 
s_\varTheta \left(\overline{Z}_t, Y; \overline{v} \right)^*
\cdot
\frac{\partial}{\partial \overline{Z}_t^*}
\log{
p_{\mathsf{Y} \mid \overline{\mathsf{Z}}_t }
\left(Y \mid \overline{Z}_t \right)}
\right] \right)^*
\nonumber
\\
&
\quad
+ \, \mathrm{const},
\label{equ: e_dsm_expansion}
\end{align}
where ``$\mathrm{const}$'' denotes a term independent of $\varTheta$.

The term (a) can be expressed as an expectation with respect to $ \left( \overline{Z}_t, Z, Y \right)$ as
\begin{align}
\textrm{(a)}
&=
\int_{\mathbb{C}}
\int_{\mathbb{C}}
\left|
s_\varTheta(\overline{z}, y; \overline{v}) 
\right|^2
\,
p_{ \overline{\mathsf{Z}}_t,  \mathsf{Y} }
(\overline{z}, y)
\,
\mathrm{d} \overline{z}
\mathrm{d} y
\nonumber
\\
&
=
\int_{\mathbb{C}}
\int_{\mathbb{C}}
\int_{\mathbb{C}}
\left|
s_\varTheta(\overline{z}, y; \overline{v}) 
\right|^2
\,
p_{ \overline{\mathsf{Z}}_t, \mathsf{Z}, \mathsf{Y} }
(\overline{z}, z, y)
\,
\mathrm{d} \overline{z}
\mathrm{d} z
\mathrm{d} y
\nonumber
\\
&
=
\mathbb{E}_{ \overline{\mathsf{Z}}_t,  \mathsf{Z}, \mathsf{Y}  }
\left[ 
\left|
s_\varTheta \left(\overline{Z}_t, Y; \overline{v} \right) 
\right|^2
\right].
\label{equ: a_result}
\end{align}
For the term (b), the following Lemma~\ref{lem: (b)_dsm} will be useful.
\begin{lemma}
\label{lem: (b)_dsm}
    Let $(\overline{Z}_t, Z, Y)$ be a random triplet with their joint \ac{PDF} given by \eqref{equ: se_joint_pdf}. 
    Then, the following equality holds:
    \begin{align}
    &
    \mathbb{E}_{ \overline{\mathsf{Z}}_t, \mathsf{Y} }
    \left[ 
    s_\varTheta \left(\overline{Z}_t, Y; \overline{v} \right)^*
    \cdot
    \frac{\partial}{\partial \overline{Z}_t^*}
    \log{
    p_{\mathsf{Y} \mid \overline{\mathsf{Z}}_t }
    \left(Y \mid \overline{Z}_t \right)}
    \right]  
    \nonumber
    \\
    &=
    \mathbb{E}_{ \overline{\mathsf{Z}}_t,  \mathsf{Z}, \mathsf{Y}  }
    \left[ 
    s_\varTheta \left(\overline{Z}_t, Y; \overline{v} \right)^*
    \cdot
    \frac{ Z - \overline{Z}_t }{\overline{v}}
    \right].
    \label{equ: (b)_dsm}
\end{align} 
\end{lemma}

\begin{IEEEproof}
See Appendix~\ref{chap: app_dsm_deriv}.
\end{IEEEproof}
Therefore, by substituting \eqref{equ: a_result} and \eqref{equ: (b)_dsm} into \eqref{equ: e_dsm_expansion}, we establish the equivalence between $J_{\mathrm{SM}}$ and $J_{\mathrm{DSM}}$ up to an additive constant as
\begin{align}
&J_{\mathrm{SM}}(\varTheta; \overline{v})
\nonumber
\\
&
=
\mathbb{E}_{ \overline{\mathsf{Z}}_t,  \mathsf{Z}, \mathsf{Y}  }
\left[ 
\left|
s_\varTheta \left(\overline{Z}_t, Y; \overline{v} \right) 
\right|^2
\right]
\nonumber
\\
&
\quad
-
\mathbb{E}_{ \overline{\mathsf{Z}}_t,  \mathsf{Z}, \mathsf{Y}  }
\left[ 
s_\varTheta \left(\overline{Z}_t, Y; \overline{v} \right)^*
\cdot
\frac{ Z - \overline{Z}_t }{\overline{v}}
\right]
\nonumber
\\
&
\quad
-
\mathbb{E}_{ \overline{\mathsf{Z}}_t,  \mathsf{Z}, \mathsf{Y}  }
\left[ 
s_\varTheta \left(\overline{Z}_t, Y; \overline{v} \right)
\cdot
\frac{ \left( Z - \overline{Z}_t \right)^* }{\overline{v}}
\right] + \mathrm{const}
\nonumber
\\
&
=
\underbrace{
\mathbb{E}_{ \overline{\mathsf{Z}}_t, \mathsf{Z}, \mathsf{Y} }
\!
\left[ 
\left| 
s_{\varTheta} \left(\overline{Z}_t, Y; \overline{v}\right)
+
\frac{ \overline{Z}_t - Z }{\overline{v}}
\right|^2
\right]}_{= J_{\mathrm{DSM}}(\varTheta. \overline{v}) } + \, \mathrm{const}.
\label{equ: sm_and_dsm}
\end{align}

\subsection{Asymptotic Invariance of the Achievable MSE}
Finally, we briefly show that replacing $g_{\mathrm{out}}(\overline{z}; \overline{v}, y)$ with the trained network $s_{\varTheta^\star}(\overline{z}, y; \overline{v})$ in (A\ref{alg: gamp_sv_div}-4,~5) of Algorithm~\ref{alg: gamp_sv_div}, referred to as \emph{SE-DSM-GAMP}, asymptotically achieves the same $\mathrm{MSE}(t)$ in \eqref{equ: def_mse} for $t = 1, 2, \ldots, T$ as Algorithm~\ref{alg: gamp_sv_div} in the large-system limit.

From the definition of $\varTheta^\star$ in \eqref{equ: def_theta_star}, together with the relationship between $J_{\mathrm{SM}}$ and $J_{\mathrm{DSM}}$ in \eqref{equ: sm_and_dsm} and the definition of $J_{\mathrm{SM}}$ in \eqref{equ: e_dsm}, the optimal parameter $\varTheta=\varTheta^\star$ satisfies
\begin{equation}
\mathbb{E}_{\overline{\mathsf{Z}}_t, \mathsf{Y}}
\left[
\left|
s_{\varTheta^\star} (\overline{Z}_t, Y; \overline{v}) - 
g_{\mathrm{out}} \left(\overline{Z}_t; \overline{v}, Y \right)
\right|^2
\right]
= 0,
\end{equation}
for a sufficiently rich network $s_{\varTheta}$.
This equivalently means that the equality
\begin{equation}
s_{\varTheta^\star} (\overline{z}, y; \overline{v}) = 
g_{\mathrm{out}} \left(\overline{z}; \overline{v}, y \right)
\end{equation}
holds almost everywhere on the support of the joint distribution in \eqref{equ: se_joint_pdf}.
Therefore, for SE-DSM-GAMP, applying the standard \ac{SE} analysis for a general output decision function\footnote{Note that the \ac{SE} analysis remains valid for SE-DSM-GAMP with an arbitrary Lipschitz-continuous network $s_{\varTheta}$, but it generally leads to a different (generalized) \ac{SE} recursion than that of Algorithm~\ref{alg: gamp_se}.} yields the same \ac{SE} recursion as in Algorithm~\ref{alg: gamp_se}.
Consequently, from \eqref{equ: def_empirical_convergence}, \eqref{equ: se_conv_in}, and \eqref{equ: div_conv}, the asymptotic MSE
\begin{align}
\mathrm{MSE}(t)
&=
\lim_{N = \xi M \to \infty}
\frac{1}{M} \left\| \hat{\bm{x}}(t+1) - \bm{x} \right\|^2 
\nonumber
\\
&=
\lim_{N = \xi M \to \infty}
\frac{1}{M} \left\| g_{\mathrm{in}} \left( \overline{\bm{x}}(t); -\frac{1}{\alpha^{\mathrm{out}}(t)} \right) - \bm{x} \right\|^2 
\nonumber
\\
&
\overset{\mathrm{a.s.}}{=}
\mathbb{E}_{\mathsf{X}, \overline{\mathsf{X}}_{t}}
\left[ 
\left| X - g_{\mathrm{in}} \left( \overline{X}_{t}; - \frac{1}{\overline{\alpha}_{t}^{\mathrm{out}}} \right) \right|^2
\right]
\end{align}
remains unchanged between the two algorithms, thereby establishing Proposition~\ref{prop: se_dsm}.

%% file: ALG/gamp_se.tex


\begin{algorithm*}[!t]
\caption{ State Evolution Recursion of Complex-Valued Bayes-GAMP}
\label{alg: gamp_se}
\begin{algorithmic}[1]


%
%
\For{$t=1 \,\, \textrm{to} \,\, T$}

\If {$ t = 1 $}
\\
\quad
$ Z \sim \mathcal{CN} \left(0, \xi^{-1} \cdot \sigma_{\mathrm{x}}^2 \right), \,\, \overline{Z}_1 = 0, \,\, 
-\overline{\alpha}_1^{\mathrm{in}}/\overline{\alpha}_0^{\mathrm{out}} = \sigma_{\mathrm{x}}^2. $
\Else \\
\quad
$ 
\displaystyle
\begin{bmatrix}
Z \\
\overline{Z}_t \\
\end{bmatrix}
\sim
\mathcal{CN}
\left( 
\begin{bmatrix}
0 \\
0 \\
\end{bmatrix},
\frac{1}{\xi} 
\begin{bmatrix}
\sigma_{\mathrm{x}}^2 & 
\mathbb{E}_{\overline{\mathsf{X}}_{t-1}} \left[ \left|g_{\mathrm{in}} \left(\overline{X}_{t-1}; -\frac{1}{\overline{\alpha}_{t-1}^{\mathrm{out}}} \right)\right|^2 \right]
\\
\mathbb{E}_{\overline{\mathsf{X}}_{t-1}} \left[ \left|g_{\mathrm{in}} \left(\overline{X}_{t-1}; -\frac{1}{\overline{\alpha}_{t-1}^{\mathrm{out}}} \right)\right|^2 \right]
&
\mathbb{E}_{\overline{\mathsf{X}}_{t-1}} \left[ \left|g_{\mathrm{in}} \left(\overline{X}_{t-1}; -\frac{1}{\overline{\alpha}_{t-1}^{\mathrm{out}}} \right)\right|^2 \right] \\
\end{bmatrix}
\right) $
\EndIf

\State $
\displaystyle
\overline{\alpha}_t^{\mathrm{out}}
=
\mathbb{E}_{\mathsf{Z}, \overline{\mathsf{Z}}_t, \mathsf{W}}
\left[
g_{\mathrm{out}}'
\left(
\overline{Z}_t;
-
\frac{1}{\xi}
\cdot
\frac{\overline{\alpha}^{\mathrm{in}}_t}{\overline{\alpha}^{\mathrm{out}}_{t-1}},
h \left( Z, W \right)
 \right)
\right] $

\State $ \displaystyle
\overline{X}_t = X + H_t,
\quad
H_t \sim \mathcal{CN}\left( 0, -1/\overline{\alpha}_t^{\mathrm{out}} \right)
$
\State $ \displaystyle
\overline{\alpha}_{t+1}^{\mathrm{in}}
=
\mathbb{E}_{\overline{\mathsf{X}}_t}
\left[
g_{\mathrm{in}}'
\left(
\overline{X}_t;
-
1/\overline{\alpha}^{\mathrm{out}}_{t}
\right)
\right] $
\EndFor
\end{algorithmic}
\end{algorithm*}

%% file: TXT/conc.tex
In this paper, we proposed a novel framework, termed \emph{SE-DSM}, for emulating the output decision function of Bayes-\ac{GAMP} for complex-valued \acp{GLM}.
Although Bayes-\ac{GAMP} has strong theoretical appeal, including \ac{SE}-based characterization of its message dynamics and asymptotic Bayes-optimality guarantees, its closed-form implementation is possible only for a limited class of observation models.
This limitation stems mainly from the restrictive requirement that the score function associated with the nonlinear observation mapping $h$, and hence the corresponding \ac{MMSE} output denoiser, admit a closed-form expression.
To bridge the gap between such theoretical guarantees and practical implementation, we developed a principled methodology that combines \ac{SE}-based asymptotic analysis with data-driven function approximation.
Specifically, SE-DSM enables the output decision function to be learned even when it is analytically intractable, requiring neither an explicit functional form of $h$ nor explicit evaluation of the true score function; only forward evaluations of $h$ are needed during offline training.
We also proved that, under ideal SE-DSM training, \ac{GAMP} equipped with the trained network asymptotically achieves the same performance as Bayes-\ac{GAMP} equipped with the exact but analytically intractable output decision function.
These results provide a practical route to extending Bayes-\ac{GAMP} to substantially broader classes of complex-valued nonlinear observation models while preserving its asymptotic performance.

%% file: TXT/appendix.tex
\section{Proof of Lemma~\ref{lem: (b)_dsm}}
\label{chap: app_dsm_deriv}

By expanding the expectation as an explicit integral, \eqref{equ: (b)_dsm} can be rewritten as \eqref{equ: dsm_b_result},
\begin{figure*}[t]
\begin{align}
&
\nonumber
\mathbb{E}_{ \overline{\mathsf{Z}}_t, \mathsf{Y} }
    \left[ 
    s_\varTheta \left(\overline{Z}_t, Y; \overline{v} \right)^*
    \cdot
    \frac{\partial}{\partial \overline{Z}_t^*}
    \log{
    p_{\mathsf{Y} \mid \overline{\mathsf{Z}}_t }
    \left(Y \mid \overline{Z}_t \right)}
    \right]
\\
&\quad=
\int_{\mathbb{C}}
\int_{\mathbb{C}}
s_\varTheta(\overline{z}, y; \overline{v})^*
\left[
\frac{\partial}{\partial \overline{z}^*}
\log{
p_{\mathsf{Y} \mid \overline{\mathsf{Z}}_t }
\left(y \mid \overline{z} \right)}
\right]
p_{ \overline{\mathsf{Z}}_t,  \mathsf{Y} }
(\overline{z}, y)
\,
\mathrm{d} \overline{z}
\mathrm{d} y
\nonumber
\\
&
\quad=
\int_{\mathbb{C}}
\int_{\mathbb{C}}
s_\varTheta(\overline{z}, y; \overline{v})^*
\cdot
\frac{
\frac{\partial}{\partial \overline{z}^*}
p_{\mathsf{Y} \mid \overline{\mathsf{Z}}_t }
\left(y \mid \overline{z} \right)
}{
p_{\mathsf{Y} \mid \overline{\mathsf{Z}}_t }
\left(y \mid \overline{z} \right)
}
\cdot
p_{ \overline{\mathsf{Z}}_t }
(\overline{z})
\cdot
p_{  \mathsf{Y} \mid \overline{\mathsf{Z}}_t  }
(y \mid \overline{z})
\,
\mathrm{d} \overline{z}
\mathrm{d} y
\nonumber
\\
&
\quad=
\int_{\mathbb{C}}
\int_{\mathbb{C}}
s_\varTheta(\overline{z}, y; \overline{v})^*
\left[
\frac{\partial}{\partial \overline{z}^*}
p_{\mathsf{Y} \mid \overline{\mathsf{Z}}_t }
\left(y \mid \overline{z} \right)
\right]
p_{ \overline{\mathsf{Z}}_t }
(\overline{z})
\,
\mathrm{d} \overline{z}
\mathrm{d} y
\nonumber
\\
&
\quad=
\int_{\mathbb{C}}
\int_{\mathbb{C}}
s_\varTheta(\overline{z}, y; \overline{v})^*
\left[
\frac{\partial}{\partial \overline{z}^*}
\int_{\mathbb{C}}
p_{\mathsf{Z}, \mathsf{Y} \mid \overline{\mathsf{Z}}_t }
\left(z, y \mid \overline{z} \right)
\, \mathrm{d}z
\right]
p_{ \overline{\mathsf{Z}}_t }
(\overline{z})
\,
\mathrm{d} \overline{z}
\mathrm{d} y
\nonumber
\\
&
\quad=
\int_{\mathbb{C}}
\int_{\mathbb{C}}
s_\varTheta(\overline{z}, y; \overline{v})^*
\left[
\frac{\partial}{\partial \overline{z}^*}
\!\!
\int_{\mathbb{C}}
p_{\mathsf{Z} \mid \overline{\mathsf{Z}}_t }
\left(z \mid \overline{z} \right)
p_{\mathsf{Y} \mid \mathsf{Z} }
\left(y \mid z \right)
\, \mathrm{d}z
\right]
p_{ \overline{\mathsf{Z}}_t }
(\overline{z})
\,
\mathrm{d} \overline{z}
\mathrm{d} y
\nonumber
\\
&
\quad=
\int_{\mathbb{C}}
\int_{\mathbb{C}}
s_\varTheta(\overline{z}, y; \overline{v})^*
\left[
\int_{\mathbb{C}}
\left[
\frac{\partial}{\partial \overline{z}^*}
p_{\mathsf{Z} \mid \overline{\mathsf{Z}}_t }
\left(z \mid \overline{z} \right)
\right]
p_{\mathsf{Y} \mid  \mathsf{Z} }
\left(y \mid  z \right)
\, \mathrm{d}z
\right]
p_{ \overline{\mathsf{Z}}_t }
(\overline{z})
\,
\mathrm{d} \overline{z}
\mathrm{d} y
\nonumber
\\
&
\quad=
\int_{\mathbb{C}}
\int_{\mathbb{C}}
\int_{\mathbb{C}}
s_\varTheta(\overline{z}, y; \overline{v})^*
\left[
\frac{\partial}{\partial \overline{z}^*}
p_{\mathsf{Z} \mid \overline{\mathsf{Z}}_t }
\left(z \mid \overline{z} \right)
\right]
p_{\mathsf{Y} \mid \mathsf{Z} }
\left(y \mid z \right)
p_{ \overline{\mathsf{Z}}_t }
(\overline{z})
\,
\mathrm{d} \overline{z}
\mathrm{d}z
\mathrm{d} y
\nonumber
\\
&
\quad=
\int_{\mathbb{C}}
\int_{\mathbb{C}}
\int_{\mathbb{C}}
s_\varTheta(\overline{z}, y; \overline{v})^*
\cdot
\frac{
\frac{\partial}{\partial \overline{z}^*}
p_{\mathsf{Z} \mid \overline{\mathsf{Z}}_t }
\left(z \mid \overline{z} \right)}
{ p_{\mathsf{Z} \mid \overline{\mathsf{Z}}_t }
\left(z \mid \overline{z} \right) }
\cdot
\underbrace{
p_{ \overline{\mathsf{Z}}_t }
(\overline{z})
p_{\mathsf{Z} \mid \overline{\mathsf{Z}}_t }
\left(z \mid \overline{z} \right)
p_{\mathsf{Y} \mid \mathsf{Z} }
\left(y \mid z \right)}_{= p_{\overline{\mathsf{Z}}_t, \mathsf{Z}, \mathsf{Y}} (\overline{z}, z, y) \textrm{ from \eqref{equ: se_joint_pdf}} }
\,
\mathrm{d} \overline{z}
\mathrm{d}z
\mathrm{d} y
\nonumber
\\
&
\quad=
\int_{\mathbb{C}}
\int_{\mathbb{C}}
\int_{\mathbb{C}}
s_\varTheta(\overline{z}, y; \overline{v})^*
\left[ 
\frac{\partial}{\partial \overline{z}^*}
\log{ p_{\mathsf{Z} \mid \overline{\mathsf{Z}}_t }
\left(z \mid \overline{z} \right) }
\right]
p_{\overline{\mathsf{Z}}_t, \mathsf{Z}, \mathsf{Y}} (\overline{z}, z, y)
\,
\mathrm{d} \overline{z}
\mathrm{d}z
\mathrm{d} y
\nonumber
\\
&
\quad=
\mathbb{E}_{ \overline{\mathsf{Z}}_t,  \mathsf{Z}, \mathsf{Y}  }
\left[ 
s_\varTheta \left(\overline{Z}_t, Y; \overline{v} \right)^*
\cdot
\frac{\partial}{\partial \overline{Z}_t^*}
\log{ p_{\mathsf{Z} \mid \overline{\mathsf{Z}}_t }
\left(Z \mid \overline{Z}_t \right) }
\right]
=
\mathbb{E}_{ \overline{\mathsf{Z}}_t,  \mathsf{Z}, \mathsf{Y}  }
\left[ 
s_\varTheta \left(\overline{Z}_t, Y; \overline{v} \right)^*
\cdot
\frac{ Z - \overline{Z}_t }{\overline{v}}
\right].
\label{equ: dsm_b_result}
\end{align}
\hrule
\end{figure*}
where the last equality in \eqref{equ: dsm_b_result} follows from 
taking the logarithm of
$ p_{\mathsf{Z} \mid \overline{\mathsf{Z}}_t} (z \mid \overline{z}) = (\pi \overline{v})^{-1} \cdot e^{- \frac{|z - \overline{z}|^2}{\overline{v}}} $ in \eqref{equ: se_joint_pdf}
and then differentiating with respect to $\overline{z}^*$, \textit{i.e.},
\begin{align}
&
\frac{\partial}{\partial \overline{z}^*}
\log{
p_{\mathsf{Z} \mid \overline{\mathsf{Z}}_t }
\left(z \mid \overline{z} \right)}
 =
\frac{\partial}{\partial \overline{z}^*}
\left[
- \frac{ |z - \overline{z}|^2 }{\overline{v}} - \log{(\pi \overline{v})} \right]
\nonumber
\\
& \!\!\!\!\!\! =
-
\frac{\partial}{\partial \overline{z}^*}
\left[
\frac{(\overline{z}^* - z^*) (\overline{z} - z)}{\overline{v}}
\right]
=
\frac{z - \overline{z}}{\overline{v}}.
\end{align}
%

%% file: REF/conf_abbrv.bib
@string{ isit = {Proc. ISIT}}


%% file: REF/ref.bib
@INPROCEEDINGS{Rangan2011,
author={S. {Rangan}},
booktitle={Proc. 2011 IEEE Int. Symp. Inf. Theory (ISIT)},
title={Generalized approximate message passing for estimation with random linear mixing},
year={2011},
volume={},
number={},
pages={2168-2172},
doi={10.1109/ISIT.2011.6033942},
ISSN={2157-8095},
month={July},
}

@article{Donoho2009,
title={Message-passing algorithms for compressed sensing},
author={Donoho, David L and Maleki, Arian and Montanari, Andrea},
journal={Proc. of the National Academy of Sciences},
volume={106},
number={45},
pages={18914--18919},
year={2009},
publisher={National Acad Sciences}
}

@ARTICLE{Wen2016,
  author={C. {Wen} and C. {Wang} and S. {Jin} and K. {Wong} and P. {Ting}},
  journal={IEEE Trans. Signal Process.}, 
  title={{B}ayes-Optimal Joint Channel-and-Data Estimation for Massive {MIMO} With Low-Precision {ADC}s}, 
  year={2016},
  volume={64},
  number={10},
  pages={2541-2556},
  }

@ARTICLE{Schniter2015,
  author={P. {Schniter} and S. {Rangan}},
  journal={IEEE Trans. Signal Processing}, 
  title={Compressive Phase Retrieval via Generalized Approximate Message Passing}, 
  year={2015},
  volume={63},
  number={4},
  pages={1043-1055},
  }

@ARTICLE{Bayati2011,
  author={Bayati, Mohsen and Montanari, Andrea},
  journal={IEEE Trans. Inf. Theory}, 
  title={The Dynamics of Message Passing on Dense Graphs, with Applications to Compressed Sensing}, 
  year={2011},
  volume={57},
  number={2},
  pages={764-785},
  doi={10.1109/TIT.2010.2094817}}

@ARTICLE{Meng2018,
  author={Meng, Xiangming and Wu, Sheng and Zhu, Jiang},
  journal={IEEE Signal Process. Lett.}, 
  title={A Unified {B}ayesian Inference Framework for Generalized Linear Models}, 
  year={2018},
  volume={25},
  number={3},
  pages={398-402},
  doi={10.1109/LSP.2017.2789163}}

@ARTICLE{Kschischang2001,
  author={Kschischang, F.R. and Frey, B.J. and Loeliger, H.-A.},
  journal={IEEE Trans. Inf. Theory}, 
  title={Factor graphs and the sum-product algorithm}, 
  year={2001},
  volume={47},
  number={2},
  pages={498-519},
  doi={10.1109/18.910572}}

@ARTICLE{Ma2017,
  author={Ma, Junjie and Ping, Li},
  journal={IEEE Access}, 
  title={Orthogonal {AMP}}, 
  year={2017},
  volume={5},
  number={},
  pages={2020-2033},
  doi={10.1109/ACCESS.2017.2653119}}

@ARTICLE{Rangan2019,
  author={Rangan, Sundeep and Schniter, Philip and Fletcher, Alyson K.},
  journal={IEEE Trans. Inf. Theory}, 
  title={Vector Approximate Message Passing}, 
  year={2019},
  volume={65},
  number={10},
  pages={6664-6684},
  doi={10.1109/TIT.2019.2916359}}

@ARTICLE{Takeuchi2020,
  author={Takeuchi, Keigo},
  journal={IEEE Trans. Inf. Theory}, 
  title={Rigorous Dynamics of Expectation-Propagation-Based Signal Recovery from Unitarily Invariant Measurements}, 
  year={2020},
  volume={66},
  number={1},
  pages={368-386},
  doi={10.1109/TIT.2019.2947058}}

@ARTICLE{Kamilov2012,
  author={Kamilov, Ulugbek S. and Goyal, Vivek K and Rangan, Sundeep},
  journal={IEEE Trans. Signal Process.}, 
  title={Message-Passing De-Quantization With Applications to Compressed Sensing}, 
  year={2012},
  volume={60},
  number={12},
  pages={6270-6281},
  doi={10.1109/TSP.2012.2217334}}

@ARTICLE{Yang2020,
  author={Yang, Xi and Wen, Chao-Kai and Jin, Shi and Swindlehurst, A. Lee},
  journal={IEEE Trans. Signal Process.}, 
  title={{B}ayes-Optimal {MMSE} Detector for Massive {MIMO} Relaying With Low-Precision {ADC}s/{DAC}s}, 
  year={2020},
  volume={68},
  number={},
  pages={3341-3357},
  doi={10.1109/TSP.2020.2977265}}

@ARTICLE{Yang2021,
  author={Yang, Chaosan and Liu, Xiaobei and Guan, Yong Liang and Liu, Rongke},
  journal={IEEE Commun. Lett.}, 
  title={Fast {GAMP} Algorithm for Nonlinearly Distorted {OFDM} Signals}, 
  year={2021},
  volume={25},
  number={5},
  pages={1682-1686},
  doi={10.1109/LCOMM.2021.3056657}}

@INPROCEEDINGS{Schniter2016,
  author={Schniter, Philip and Rangan, Sundeep and Fletcher, Alyson K.},
  booktitle={Proc. Conf. Rec. Asilomar Conf. Signals Syst.}, 
  title={Vector approximate message passing for the generalized linear model}, 
  year={2016},
  volume={},
  number={},
  pages={1525-1529},
  doi={10.1109/ACSSC.2016.7869633}}

@ARTICLE{Zou2018,
  author={Zou, Qiuyun and Zhang, Haochuan and Wen, Chao-Kai and Jin, Shi and Yu, Rong},
  journal={IEEE Signal Process. Lett.}, 
  title={Concise Derivation for Generalized Approximate Message Passing Using Expectation Propagation}, 
  year={2018},
  volume={25},
  number={12},
  pages={1835-1839},
  doi={10.1109/LSP.2018.2876806}}

@ARTICLE{Liu2019tvt,
  author={Liu, Lei and Li, Ying and Huang, Chongwen and Yuen, Chau and Guan, Yong Liang},
  journal={IEEE Trans. Veh. Technol.}, 
  title={A New Insight Into {GAMP} and {AMP}}, 
  year={2019},
  volume={68},
  number={8},
  pages={8264-8269},
  doi={10.1109/TVT.2019.2926229}}

@article{
Barbier2019,
author = {Jean Barbier  and Florent Krzakala  and Nicolas Macris  and Léo Miolane  and Lenka Zdeborová },
title = {Optimal errors and phase transitions in high-dimensional generalized linear models},
journal = {Proc. Natl. Acad. Sci.},
volume = {116},
number = {12},
pages = {5451-5460},
year = {2019},
doi = {10.1073/pnas.1802705116},
URL = {https://www.pnas.org/doi/abs/10.1073/pnas.1802705116},
eprint = {https://www.pnas.org/doi/pdf/10.1073/pnas.1802705116}}

@inproceedings{Pearl1982,
author = {Pearl, Judea},
title = {Reverend {B}ayes on inference engines: a distributed hierarchical approach},
year = {1982},
publisher = {AAAI Press},
booktitle = {Proc. 2nd AAAI Conf. Artif. Intell.},
pages = {133–136},
numpages = {4},
location = {Pittsburgh, Pennsylvania},
series = {AAAI'82}
}

@ARTICLE{Takeuchi2025,
  author={Takeuchi, Keigo},
  journal={IEEE Trans. Inf. Theory}, 
  title={Generalized Approximate Message-Passing for Compressed Sensing With Sublinear Sparsity}, 
  year={2025},
  volume={71},
  number={6},
  pages={4602-4636},
  doi={10.1109/TIT.2025.3560070}}

@ARTICLE{Chi2024,
  author={Chi, Yuhao and Chen, Xuehui and Liu, Lei and Li, Ying and Bai, Baoming and Al Hammadi, Ahmed and Yuen, Chau},
  journal={IEEE Trans. Commun.}, 
  title={{GAMP} or {GOAMP}/{GVAMP} Receiver in Generalized Linear Systems: Achievable Rate, Coding Principle, and Comparative Study}, 
  year={2024},
  volume={72},
  number={10},
  pages={6237-6253},
  doi={10.1109/TCOMM.2024.3395337}}

@ARTICLE{Javanmard2013,
  author={Javanmard, Adel and Montanari, Andrea},
  journal={Inf. Inference: A Journal of the IMA}, 
  title={State evolution for general approximate message passing algorithms, with applications to spatial coupling}, 
  year={2013},
  volume={2},
  number={2},
  pages={115-144},
  doi={10.1093/imaiai/iat004}}

@book{Pace1997book,
  title={Principles Of Statistical Inference From A Neo-fisherian Perspective},
  author={Pace, L. and Salvan, A.},
  isbn={9789813103016},
  series={Advanced Series On Statistical Science And Applied Probability},
  year={1997},
  publisher={World Scientific Publishing Company}
}

@misc{Rangan2012arxiv,
      title={Generalized Approximate Message Passing for Estimation with Random Linear Mixing}, 
      author={Sundeep Rangan},
      year={2012},
      eprint={1010.5141},
      archivePrefix={arXiv},
      primaryClass={cs.IT},
      url={https://arxiv.org/abs/1010.5141}, 
}

@ARTICLE{Zhu2019spl,
  author={Zhu, Jiang and Yuan, Qiumeng and Song, Chunyi and Xu, Zhiwei},
  journal={IEEE Signal Process. Lett.}, 
  title={Phase Retrieval From Quantized Measurements via Approximate Message Passing}, 
  year={2019},
  volume={26},
  number={7},
  pages={986-990},
  doi={10.1109/LSP.2019.2916668}}

@ARTICLE{Ma2019TIT,
  author={Ma, Junjie and Xu, Ji and Maleki, Arian},
  journal={IEEE Trans. Inf. Theory}, 
  title={Optimization-Based {AMP} for Phase Retrieval: The Impact of Initialization and  $\ell_{2}$  Regularization}, 
  year={2019},
  volume={65},
  number={6},
  pages={3600-3629},
  doi={10.1109/TIT.2019.2893254}}

@INPROCEEDINGS{Zhidkov2019vtc,
  author={Zhidkov, Sergey V. and Dinis, Rui},
  booktitle = {Proc. IEEE Veh. Technol. Conf. (VTC-Spring)}, 
  title={Belief Propagation Receivers for Near-Optimal Detection of Nonlinearly Distorted {OFDM} Signals}, 
  year={2019},
  volume={},
  number={},
  pages={1-6},
  doi={10.1109/VTCSpring.2019.8746319}}

@ARTICLE{Vincent2011,
  author={Vincent, Pascal},
  journal={Neural Comput.}, 
  title={A Connection Between Score Matching and Denoising Autoencoders}, 
  year={2011},
  volume={23},
  number={7},
  pages={1661-1674},
  doi={10.1162/NECO_a_00142}}

@article{Hyvarinen2005,
  author  = {A. Hyv{\"a}rinen},
  title   = {Estimation of Non-Normalized Statistical Models by Score Matching},
  journal = {J. Mach. Learn. Res.},
  volume    = {6},
  number     = {24},
  pages     = {695--709},
  year    = {2005}
}

@article{Feng2022,
  author  = {O. Y. Feng and R. Venkataramanan and C. Rush and R. J. Samworth},
  title   = {A Unifying Tutorial on Approximate Message Passing},
  journal = {Found. Trends Mach. Learn.},
  volume    = {15},
  number     = {4},
  pages     = {335--536},
  year    = {2022},
  doi     = {10.1561/2200000092}
}

@misc{Wadayama2026scvamp,
      title={Score-Based {VAMP} with {Fisher}-Information-Based {Onsager} Correction}, 
      author={Tadashi Wadayama and Takumi Takahashi},
      year={2026},
      eprint={2601.07095},
      archivePrefix={arXiv},
      primaryClass={cs.IT},
      url={https://arxiv.org/abs/2601.07095}, 
}

@misc{Wadayama2026tmvamp,
      title={Three-Module {SC}-{VAMP} for {LDPC}-Coded Nonlinear Channels}, 
      author={Tadashi Wadayama and Takumi Takahashi},
      year={2026},
      eprint={2604.19061},
      archivePrefix={arXiv},
      primaryClass={cs.IT},
      url={https://arxiv.org/abs/2604.19061}, 
}

@inproceedings{Song2019,
  author    = {Y. Song and S. Ermon},
  title     = {Generative Modeling by Estimating Gradients of the Data Distribution},
  booktitle = {Adv. Neural Inf. Process. Syst. (NeurIPS)},
  pages = {11895–11907},
  volume      = {32},
  year      = {2019}
}

@article{Mondelli2022,
  author  = {M. Mondelli and R. Venkataramanan},
  title   = {Approximate Message Passing With Spectral Initialization for Generalized Linear Models},
  journal = {J. Stat. Mech.: Theory Exp.},
  volume  = {2022},
  number  = {11},
  pages   = {114003},
  month   = nov,
  year    = {2022},
  doi     = {10.1088/1742-5468/ac9828}
}

@book{Kay1993,
  author    = {Steven M. Kay},
  title     = {Fundamentals of Statistical Signal Processing: Estimation Theory},
  volume     = {1},
  publisher = {Prentice-Hall PTR},
  address   = {Englewood Cliffs, NJ},
  year      = {1993}
}

@INPROCEEDINGS{Cai2025spawc,
  author={Cai, Chang and Yuan, Xiaojun and Zhang, Ying-Jun Angela},
  booktitle={Proc. IEEE Int. Workshop Signal Process. Artif. Intell. Wireless Commun. (SPAWC)}, 
  title={Score-Based Turbo Message Passing for Plug-and-Play Compressive Image Recovery}, 
  year={2025},
  volume={},
  number={},
  pages={1-5},
  doi={10.1109/SPAWC66079.2025.11143463}}

@ARTICLE{Yu2024,
  author={Yu, Wentao and He, Hengtao and Yu, Xianghao and Song, Shenghui and Zhang, Jun and Murch, Ross and Letaief, Khaled B.},
  journal={IEEE J. Sel. Top. Signal Process.}, 
  title={Bayes-Optimal Unsupervised Learning for Channel Estimation in Near-Field Holographic {MIMO}}, 
  year={2024},
  volume={18},
  number={4},
  pages={714-729},
  doi={10.1109/JSTSP.2024.3414137}}

@ARTICLE{Efron2011,
  author  = {B. Efron},
  title   = {Tweedie's Formula and Selection Bias},
  journal = {J. Amer. Statist. Assoc.},
  volume  = {106},
  number  = {496},
  pages   = {1602--1614},
  year    = {2011},
  doi     = {10.1198/jasa.2011.tm11181}
}

@INPROCEEDINGS{Robbins1956,
  author    = {H. E. Robbins},
  title     = {An Empirical {Bayes} Approach to Statistics},
  booktitle = {Proc. Third Berkeley Symp. Math. Statist. Probab., Vol. 1},
  address   = {Berkeley, CA, USA},
  pages     = {157--163},
  year      = {1956}
}
